%% file: l-coverage.tex
\documentclass[11pt]{article}

\usepackage{style}
\usepackage{mathrsfs}
\usepackage{bbm}
\usepackage{thmtools}
\usepackage{thm-restate}

\declaretheorem[name=Lemma,numberwithin=section]{lem}

\newcommand{\ex}[1]{\E\big[#1\big]}
\newcommand{\pr}[1]{\Pr\big[#1\big]}
\newcommand{\prc}[2]{\Pr_{#1}\big[#2\big]}
\newcommand{\exc}[2]{\E_{#1}\big[#2\big]}

\newcommand{\RR}{\mathbb{R}}

\newcommand{\kary}{\textsc{ary-ugc}}

\newcommand{\ug}{\textsc{UniqueGames}}
\newcommand{\fcov}{\psi}

\usepackage{xcolor}

\begin{document}

\title{{\bfseries Tight Approximation Bounds for \\
Maximum Multi-Coverage}\footnote{An extended abstract appeared in the proceedings of IPCO 2020}}

\author{Siddharth Barman\thanks{Indian Institute of Science. {\tt barman@iisc.ac.in}} \and Omar Fawzi\thanks{Univ Lyon, ENS Lyon, UCBL, CNRS,  Inria, LIP, F-69342, Lyon Cedex 07, France. {\tt omar.fawzi@ens-lyon.fr}} \and Suprovat Ghoshal\thanks{Indian Institute of Science. {\tt suprovat@iisc.ac.in}} \and Emirhan G\"urp\i nar \thanks{ENS Lyon. {\tt emirhan.gurpinar@ens-lyon.fr}}}

\date{}

\maketitle

\begin{abstract}
In the classic maximum coverage problem, we are given subsets $T_1, \dots, T_m$ of a universe $[n]$ along with an integer $k$ and the objective is to find a subset $S \subseteq [m]$ of size $k$ that maximizes $C(S) \coloneqq |\cup_{i \in S} T_i|$. It is well-known that the greedy algorithm for this problem achieves an approximation ratio of $(1-e^{-1})$ and there is a matching inapproximability result. We note that in the maximum coverage problem if an element $e \in [n]$ is covered by several sets, it is still counted only once. By contrast, if we change the problem and count each element $e$ as many times as it is covered, then we obtain a linear objective function, $C^{(\infty)}(S) = \sum_{i \in S} |T_i|$, which can be easily maximized under a cardinality constraint.  

We study the maximum $\ell$-multi-coverage problem which naturally interpolates between these two extremes. In this problem, an element can be counted up to $\ell$ times but no more; hence, we consider maximizing the function $C^{(\ell)}(S) = \sum_{e \in [n]} \min\{\ell, |\{i \in S : e \in T_i\}| \}$, subject to the constraint $|S| \leq k$. Note that the case of $\ell = 1$ corresponds to the standard maximum coverage setting and $\ell = \infty$ gives us a linear objective. 

We develop an efficient approximation algorithm that achieves an approximation ratio of $1 - \frac{\ell^{\ell}e^{-\ell}}{\ell!}$ for the $\ell$-multi-coverage problem. In particular, when $\ell = 2$, this factor is $1-2e^{-2} \approx 0.73$ and as $\ell$ grows the approximation ratio behaves as $1 - \frac{1}{\sqrt{2\pi \ell}}$. We also prove that this approximation ratio is tight, i.e., establish a matching hardness-of-approximation result, under the Unique Games Conjecture. 

This problem is motivated by the question of finding a code that optimizes the list-decoding success probability for a given noisy channel. We show how the multi-coverage problem can be relevant in other contexts, such as combinatorial auctions.
\end{abstract}

\section{Introduction}

Coverage problems lie at the core of combinatorial optimization and have been extensively studied in computer science. A quintessential example of such problems is the \emph{maximum coverage} problem wherein we are given subsets $T_1, \dots, T_m$ of a universe $[n]$ along with an integer $k \in \mathbbm{Z}_+$, and the objective is to find a size-$k$ set $S \subseteq [m]$ that maximizes the covering function $C(S) \coloneqq | \cup_{i \in S} T_i |$. It is well-known that a natural greedy algorithm achieves an approximation ratio of $1-e^{-1}$ for this problem (see, e.g.,~\cite{Hoc97}). Furthermore, the work of Feige~\cite{Fei98} shows that this approximation factor is tight, under the assumption that ${\rm P} \neq {\rm NP}$. Over the years, a large body of work has been directed towards extending these fundamental results and, more generally, coverage problems have been studied across multiple fields, such as operations research~\cite{CFN77}, machine learning~\cite{feldman2014learning}, and algorithmic game theory~\cite{DV11}.

In this paper, we study the $\ell$-multi-coverage ($\ell$-coverage for short) problem, which is a natural generalization of the classic maximum coverage problem. Here, we are given a universe of elements $[n]$ and a collection of subsets $\mathcal{F} = \{T_i \subseteq [n] \}_{i=1}^m $. For any integer $\ell \in \mathbbm{Z}_+$ and a choice of index set $S \subseteq [m]$, we define the $\ell$-coverage of an element $e$ to be $C^{(\ell)}_e(S) \coloneqq \min\{\ell, | i \in S : e \in T_i |\}$, i.e., $C^{(\ell)}_e(S)$ counts---up to $\ell$---how many times element $e$ is covered by the subsets indexed in $S$. We extend this definition to that of $\ell$-coverage of all the elements, $C^{(\ell)} (S) \coloneqq \sum_{e\in [n]}  C^{(\ell)}_e(S)$.


The $\ell$-multi-coverage problem is defined as follows: given a universe of elements $[n]$, a collection $\mathcal{F}$ of subsets of $[n]$ and an integer $k \le m$, find a size-$k$ subset $S \subseteq [m]$ which maximizes $C^{(\ell)}({S})$. For $\ell = 1$, it is easy to see that this reduces to the standard maximum coverage problem. 

\subsection{Our Results and Techniques} 

Our main result is a polynomial-time algorithm that achieves a tight approximation ratio for the $\ell$-multi-coverage problem, with any $\ell \geq 1$. 


\begin{theorem}
\label{theorem:algorithm}
	Let $\ell$ be a positive integer. There exists a randomized polynomial-time algorithm that takes as input an integer $n$, a set system $\mathcal{F} = \{T_i \subseteq [n] \}_{i=1}^m $ along with an integer $k \leq m$ and outputs a size-$k$ set $S \subseteq [m]$ (i.e., identifies  $k$ subsets $\{T_i \}_{i\in S}$ from $\mathcal{F}$) such that
	\begin{equation*}
	\ex{C^{(\ell)}(S)} \ge \left(1 - \frac{\ell^{\ell} e^{-\ell}}{\ell !}\right)\max_{S' \in {[m] \choose k}}C^{(\ell)}(S').
	\end{equation*} 
\end{theorem}

One way to interpret this approximation ratio $\rho_\ell \coloneqq \left(1 - \frac{\ell^{\ell} e^{-\ell}}{\ell !}\right)$ is that $\rho_{\ell} = \frac{1}{\ell} \ex{\min \{\ell, \mathrm{Poi}(\ell)\}}$, where $\mathrm{Poi}(\ell)$ denotes a Poisson random variable with rate parameter $\ell$. 

We complement Theorem~\ref{theorem:algorithm} by proving that the achieved approximation guarantee is tight, under the Unique Games Conjecture. Formally,  
\addtocounter{thm}{+1}
\begin{restatable}{thm}{hardness} \label{theorem:hardness}
		Assuming the Unique Games Conjecture, it is {\rm NP}-hard to approximate the maximum $\ell$-multi-coverage problem to within a factor greater than $\left(1 - \frac{\ell^{\ell}}{\ell!} e^{-\ell}  + \varepsilon\right)$, for any constant of $\varepsilon > 0$.
\end{restatable}

\paragraph{The Approximation Algorithm} We first observe that for the maximum multi-coverage problem the standard greedy algorithm fails: the approximation guarantee does not improve with $\ell$. As the function $C^{(\ell)}$ is a monotone, submodular function, the greedy algorithm will certainly still achieve an approximation ratio of $1-e^{-1}$. However, it is simple to construct instances wherein exactly this ratio is achieved. In fact, if $\cF$ is a collection of distinct subsets, let $\cF^{(\ell)}$ contain the same subsets as $\cF$ but each one appearing $\ell$ times. Then, it is easy to see that the greedy algorithm, when applied to $\cF^{(\ell)}$, will simply choose $\ell$ times the sets chosen by the algorithm on input $\cF$. So the greedy algorithm is not able to take advantage when we have $\ell > 1$.

Instead, we use another algorithmic idea, which is standard in the context of submodular function maximization. We consider the natural linear programming (LP) relaxation of the problem to obtain a fractional, optimal solution and apply pipage rounding to transform this fractional solution to an integral solution. Pipage rounding was first introduced by Ageev and Sviridenko~\cite{AS04_pipage} for some specific combinatorial problems, and then generalized to submodular function maximization by Calinescu,  Chekuri, P\'al and Vondr\`ak~\cite{vondrak2007submodularity, calinescu2011maximizing}. Pipage rounding is a randomized method that maps a fractional solution $x \in [0,1]^m$ into an integral one $x^{\mathrm{int}} \in \{0,1\}^m$, in a way that preserves the constraints and does not decrease the expected value of the objective function; here the fractional solution $x \in [0,1]^m$ is viewed as a (product) distribution over the index set $[m]$. The expected value of $C^{(\ell)}$ for a set chosen according to a distribution defined by $x$ is known as the multilinear extension of $C^{(\ell)}$ evaluated at $x$, and in fact, the method of pipage rounding works for any submodular function, as submodularity is sufficient to guarantee the convexity properties used in the rounding procedure~\cite{calinescu2011maximizing}. We point out that there are other randomized rounding techniques that could be used, such as swap rounding~\cite{chekuri2009dependent}. For concreteness, we use here the result~\cite[Lemma 3.5]{calinescu2011maximizing}. Another remark is that the maximum multi-coverage problem is a special case of the class of sums of weighted rank functions of matroids for which an algorithm based on rounding a linear program is given in~\cite[Section 4]{calinescu2007ipco}. In fact, the linear program and the algorithm that we analyze here is the same as the one given in~\cite{calinescu2007ipco}. Our contribution here is to show that for the maximum $\ell$-coverage problem, the approximation factor for this algorithm, which is also the ratio between the multilinear relaxation and the linear program, can be improved from $1-e^{-1}$ in ~\cite[Lemma 6]{calinescu2007ipco} to $1 - \frac{\ell^{\ell}}{\ell!} e^{-\ell}$ in our Theorem~\ref{th:main_algo}.


Hence, the core of the analysis of this algorithm is to compute the expected $\ell$-coverage, $\exc{S \sim x}{C^{(\ell)}(S)}$, and relate it to the optimal value of the linear program (which, of course, upper bounds the value of an integral, optimal solution). With a careful use of convexity, one can establish that the analytic form of this expectation corresponds to the expected value of a binomial random variable truncated at $\ell$. 

To obtain the claimed approximation ratio (which, as mentioned above, has a Poisson interpretation), one would like to use the well-known Poisson approximation for binomial distributions. However, this convergence statement is only asymptotic and thus will lead to an error term that will depend on the size of the problem instance and on the value of $\ell$. One can alternatively try to compare the two distributions using the natural notion of stochastic domination. It turns out that indeed a binomial distribution can be stochastically dominated by a Poisson distribution, but this again cannot be used in our setting for two reasons: there is a loss in terms of the underlying parameters (and, hence, this cannot lead to a tight approximation factor) and more importantly the inequality goes in the wrong direction.\footnote{Here, the Poisson distribution stochastically dominates the binomial. Hence, instead of a lower bound, we obtain an upper bound.}

The right tool for us turns out to be the notion of \emph{convex order} between distributions. It expresses the property that one distribution is more ``spread'' than the other. While this notion has found several applications in statistics, economics, and other related fields (see \cite{bookStochOrder} and references therein), to the best of our knowledge, it is not a commonly used tool in the context of analyzing approximation algorithms.
In particular, it leads to tight comparison inequalities between binomial and Poisson distributions, even in non-asymptotic regimes (see Lemma~\ref{lem:convex-order-bin-poi}). Overall, using this tool we are able to obtain optimal approximation guarantees for all values of $\ell$.


We also note that our algorithmic result directly generalizes to the weighted version of maximum $\ell$-coverage and we can replace the constraint $|S| \leq k$ by a matroid constraint $S \in \mathcal{M}$; here, $\mathcal{M}$ is any matroid that admits an efficient, optimization algorithm (equivalently, any matroid whose basis polytope admits an efficient separation oracle). To keep the exposition simple, we conform to the unweighted case and to the cardinality constraint $|S| \leq k$ and only discuss the generalization in Section~\ref{sec:generalization}.

\paragraph{Hardness Result} 
	We now give a brief description of our hardness result and the techniques used to establish it. In \cite{Fei98}, the $(1 - 1/e)$ inapproximability of the standard maximum coverage problem was shown using the tight $\ln n$ inapproximability of the \emph{set cover} problem, which in turn was obtained via a reduction from a variant of max-3-sat. However, in our setting, one cannot hope to show tight inapproximability for the maximum $\ell$-multi-coverage problem by a similar sequence of reductions. This is because, as detailed in Section~\ref{section:related-work}, the multi-coverage analogue of the set cover problem is as inapproximable as the usual set cover problem. Therefore, one cannot hope to directly reuse the arguments from Feige's reduction in order to get tight inapproximability for the maximum $\ell$-multi-coverage problem. We bypass this by developing a direct reduction to the maximum $\ell$-multi-coverage problem without going through the set cover variant. 
	
	Our reduction is from a $h$-ary hypergraph variant of $\ug$ \cite{khot2002power, khot2005unique}, which we call $h$-$\kary$. Here the constraints are given by $h$-uniform hyperedges on a vertex set $V$ with a label set $\Sigma$. A salient feature of the $h$-$\kary$, which is crucially used in our reduction, is that it involves two distinct notions of satisfied hyperedges, namely \emph{strongly} and \emph{weakly} satisfied hyperedges. A labeling $\sigma:V \mapsto \Sigma$ strongly satisfies a hyperedge $e = (v_i)_{i \in [h]}$ if all the labels project to the same alphabet, i.e., $\pi_{e,v_1}(\sigma(v_1)) = \pi_{e,v_2}(\sigma(v_2)) = \cdots = \pi_{e,v_3}(\sigma(v_h))$. We say a labelling $\sigma$ weakly satisfies the hyperedge $e$ if at least two of the projected labels match, i.e., $\pi_{e,v_i}(\sigma(v_i)) = \pi_{e,v_j}(\sigma(v_j))$ for some $i,j \in [h], i \neq j$. The equivalent of Unique Games Conjecture for these instances is the following: {\em It is NP-Hard to distinguish between whether (YES): most hyperedges can be strongly satisfied or (NO): even a small fraction of hyperedges cannot be weakly satisfied.}
	
	We employ the above variant of $\ug$ with a generalization of Feige's partitioning gadget, which has been tailored to work with the $\ell$-coverage objective $C^{(\ell)}(\cdot)$. This gadget is essentially a collection of $s$ set families $\mathcal{P}_1,\mathcal{P}_2,\ldots,\mathcal{P}_s$ over a universe $[\widehat{n}]$ satisfying (i) Each family $\mathcal{P}_i$ is a collection of sets such that each element in $[\widehat{n}]$ is covered exactly $\ell$-times i.e, it has (normalized) $\ell$-coverage $1$ (ii) Any choice of sets $S_1,S_2,\ldots,S_h$ from distinct families has  $\ell$-coverage at most $\rho_\ell $ (the target approximation ratio). We combine the $h$-$\kary$ instance with the partitioning gadget by associating each hyperedge constraint with a disjoint copy of the gadget. The construction of the set family in our reduction ensures that sets corresponding to strongly satisfied edges use property (i), whereas sets corresponding to not even weakly satisfied hyperedges use property (ii). Since in the YES case, most hyperedges can be strongly satisfied, we get that there exists a choice of sets for which the normalized $\ell$-multi-coverage is close to $1$. On the other hand, in the NO case, since most hyperedges are not even weakly satisfied, for any choice of sets, the normalized $\ell$-multi-coverage will be at most $\rho_\ell$. Combining the two cases gives us the desired inapproximability. 
	



\subsection{Related Covering Problems and Submodular Function Maximization} 
\label{section:related-work}

Another fundamental problem in the covering context is the \emph{set cover} problem: given subsets $T_1, \dots, T_m$ of a universe $[n]$, the objective is to find the set $S \subseteq [m]$ of minimal cardinality that cover all of $n$, i.e., $C(S) = \cup_{i \in S} T_i = [n]$. This is one of the first problems for which approximation algorithms were studied: Johnson~\cite{Joh74} showed that the natural greedy algorithm achieves an approximation ratio of $1+\ln n$ and much later Feige~\cite{Fei98}, building on a long line of works, established a matching inapproximability result.\footnote{See~\cite{moshkovitz2012projection} for a hardness result based on ${\rm P} \neq {\rm NP}$}. 

Along the lines of maximum coverage, one can also consider the $\ell$-version of set cover. In this version, the goal is to find the smallest set $S$ such that $C^{(\ell)}(S) = \ell n$ (this corresponds to every element being covered at least $\ell$ times). Here, with $\ell >1$, we observe an interesting dichotomy: while one achieves improved approximation guarantees for the maximum $\ell$-multi-coverage, this is not the case for set $\ell$-cover. In particular, set $\ell$-cover is essentially as hard as to approximate as the standard set cover problem. To see this, consider the instance where $\cF^{(\ell)}$ is obtained from $\cF$ by adding $\ell-1$ copies of the whole set $[n]$. Then, we have that $[n]$ can be $1$-covered with $k$ sets in $\cF$ if and only if $[n]$ can be $\ell$-covered with $k+\ell-1$ sets in $\cF^{(\ell)}$.

A well-studied generalization of the set cover problem, called the set multicover problem, requires element $e \in [n]$ to be covered at least $d_e$ times, where the demand $d_e$ is part of the input. The greedy algorithm was shown to also achieve $1+\ln n$ approximation for this problem as well~\cite{Dob82,RV98}. Even though there has been extensive research on set multicover, its variants, and applications (see e.g.,~\cite{BDS07}), we are not aware of any previous work that considers the maximum multi-coverage problem.

The problem of maximum coverage fits within the larger framework of \emph{submodular function maximization}~\cite{NWF78}. In fact, the covering function $C : 2^{[n]} \to \RR$ is submodular in the sense that it satisfies a diminishing-returns property: $C(S \cup \{i\}) - C(S) \geq C(S' \cup \{i\}) - C(S')$ for any $S \subseteq S'$ and $i \notin S'$. Nemhauser et al.~\cite{NWF78} showed that the greedy algorithm achieves the ratio $1-e^{-1}$ not only for the coverage function $C$ but for any submodular function. Submodular functions are a central object of study in combinatorial optimization and appear in a wide variety of applications; we refer the reader to~\cite{KG12} for a textbook treatment of this topic. Here, an important thread of research is that of maximizing submodular functions that have an additional structure which render them closer to linear functions. Specifically, the notion of curvature of a function was introduced by~\cite{CC84}. The curvature of a monotone submodular function $f : 2^{[m]} \to \mathbb{R}$ is a parameter $c \in [0,1]$ such that for any $S \subset [m]$ and $j \notin S$, we have $f(S \cup \{j\}) - f(S) \geq (1-c) f(\{j\})$. Note that if $c=0$, this means that $f$ is a linear function and if $c=1$, the condition is clearly satisfied.

Conforti and Cornu{\'e}jols~\cite{CC84} have shown that when the greedy algorithm is applied to a function with curvature $c$, the approximation guarantee is $\frac{1}{c}(1-e^{-c})$.  Using a different algorithm, this was later improved by Sviridenko et al.~\cite{SVW17} to a factor of approximately $1-\frac{c}{e}$. This notion of curvature does have applications in some settings (see e.g.,~\cite{SVW17} and references therein), but the requirement is too strong and does not apply to the $\ell$-coverage function $C^{(\ell)}$. In fact, if $S$ is such that the sets $T_i$ for $i \in S$ cover all the universe at least $\ell$ times, then adding another set $T_j$ will not change the function $C^{(\ell)}$. Another way to see that this condition is not adapted to our $\ell$-coverage problem is that we know that the greedy algorithm will not be able to beat the factor $1-e^{-1}$ for any value of $\ell$. We hope that this work will help in establishing a more operational way of interpolating between general submodular functions and linear functions.




\subsection{Applications}
We now briefly discuss some applications of the $\ell$-coverage problem, the main message being that for most settings where coverage is used, $\ell$-coverage has a very natural and meaningful interpretation as well. We leave the more detailed discussion of such applications for future work. 
%

Our initial motivation for studying the maximum multi-coverage problem was in understanding the complexity of finding the code for which the list-decoding success probability is optimal. More precisely, consider a noisy channel with input set $X$ and output set $Y$ that maps an input $x \in X$ to $y \in Y$ with probability $W(y|x)$. To simplify the discussion, assume that for any input $x$, the output is uniform on a set $T_x$ of size $t$, i.e., $W(y|x) = \frac{1}{t}$ if $y \in T_x$ and $0$ otherwise. We would like to send a message $m$ belonging to the set $\{1, \dots, k\}$ using this noisy channel in such a way to maximize the probability of successfully decoding the message $m$. It is elementary to see that this problem can be written as one of maximizing the quantity $ \frac{1}{t k} | \cup_{x \in S} T_x |$ over codes $S \subseteq X$ of size $k$~\cite{BF15}. Thus, the problem of finding the optimal code can be written as a covering problem, and handling general noisy channels corresponds to a weighted covering problem. This connection was exploited in~\cite{BF15} to prove tightness of the bound known as the \emph{meta-converse} in the information theory literature~\cite{PPV10} and to give limitations on the effect of quantum entanglement to decrease the communication errors. Suppose we now consider the list-decoding success probability, i.e., the receiver now decodes $y$ into a list of size $\ell$ and we deem the decoding successful if $m$ is in this list. Then the success probability can be written as: $\frac{1}{t k} \sum_{y \in Y} \min\{ \ell, | \{ x \in S : y \in T_x \}|\}$, i.e., an $\ell$-coverage function. Our main result thus shows that the code with the maximum list-decoding success probability can be approximated to a factor of $\rho_{\ell} = 1 - \frac{\ell^\ell e^{-\ell}}{\ell !}$ and it shows that the meta-converse for list-decoding is tight within the factor $\rho_{\ell}$.

The applicability of the multi-coverage can also be observed in game-theoretic settings in which the (standard) covering function is used to represent valuations of agents; see, e.g., works on combinatorial auctions~\cite{dughmi2016optimal,dobzinski2006improved}. As a stylized instantiation, consider a setup wherein the elements in the ground set represent types of goods and the given subsets correspond to bundles of goods (of different types). Assuming that, for each agent, goods of a single type are perfect substitutes of each other, one obtains valuations (defined over the bundles) that correspond to covering functions. In this context, the $\ell$-multi-coverage formulation provides an operationally-useful generalization: additional copies (of the same type of the good) are valued, till a threshold $\ell$. Indeed, our algorithmic result shows that if the diminishing-returns property does not come into effect right away, then better (compared to $1- e^{-1}$) approximation guarantees can be obtained.


\section{Approximation Algorithm for the $\ell$-Multi-Coverage Problem}


The algorithm we analyze is simple and composed of two steps (relax and round): First, we solve the natural linear programming relaxation (see \eqref{lp_relaxation}) obtaining a fractional, optimal solution $x^* \in [0,1]^m$, which satisfies $\sum_{i \in [m]} x_i^* = k$. The second step is to use \emph{pipage rounding} to find an \emph{integral} vector $x^{\mathrm{int}} \in \{0,1\}^m$ with the property that $\sum_{i \in [m]} x^{\mathrm{int}}_i = k$. This is the size-$k$ set returned by the algorithm, $S =\{i \in [m] : x^{\mathrm{int}}_i = 1 \}$. These two steps are detailed below. 

\paragraph{Step 1. Solve the Linear Programming Relaxation:} 

Specifically, we consider the following linear programming relaxation of the $\ell$-multi-coverage problem. Here, with the given collection of sets $\mathcal{F} = \{T_1, T_2, \ldots, T_m\}$, the set $\Gamma_e \coloneqq \{ i \in [m] : e \in T_i \}$ denotes the indices of $T_i$s that contain the element $e$. 
\begin{align}\label{lp_relaxation}
\begin{aligned}
& \underset{x, c}{\max}
& & \sum_{e \in [n]} c_{e} \\
& \text{subject to}
& &c_e \leq \ell \quad \forall e \in [n] \\
& 
& &c_e \leq \sum_{i \in \Gamma_e} x_i  \quad \forall e \in [n] \\
& 
& & 0 \leq x_i \leq 1 \quad \forall i \in [m] \\ 
&
& & \sum_{i=1}^m x_i = k \,.
\end{aligned}
\end{align}
In this linear program (LP), the number of variables is $n+m$ and the number of constraints is $O(n+m)$ and, hence, an optimal solution can be found in polynomial time.



\paragraph{Step 2. Round the fractional, optimal solution:} We round the computed fractional solution $x^*$ by considering the {\em multilinear extension} of the objective, and applying pipage rounding~\cite{AS04_pipage, vondrak2007submodularity, calinescu2011maximizing} on it. Formally, given any function $f : \{0,1\}^m \to \RR$, one can define the multilinear extension $F : [0,1]^m \to \RR$ by $F(x_1, \dots, x_m) \coloneqq \ex{f(X_1, \dots, X_m)}$, where $X_1, \dots, X_m \in \{0,1\}$ are independent random variables with $\pr{X_i = 1} = x_i$. 

For a submodular function $f$, one can use pipage rounding to transform, in polynomial time, any fractional solution $x \in [0,1]^m$ satisfying $\sum_{i \in [m]} x_i = k$ into a random integral vector $x^{\mathrm{int}} \in \{0,1\}^m$ such that $\sum_{i \in [m]} x^\mathrm{int}_i = k$ and $\ex{F(x^{\mathrm{int}})} \geq F(x)$. Specifically, we use~\cite[Lemma 3.5]{calinescu2011maximizing}. We apply this strategy for the $\ell$-coverage function and the fractional, optimal solution $x^*$ of the LP relaxation~\eqref{lp_relaxation}. It is simple to check that the $\ell$-coverage function $C^{(\ell)}$ is submodular. We thus get the following lower bound for the $\ell$-coverage value of the set returned by the algorithm:\footnote{That is, a lower bound for the $\ell$-coverage value of the size-$k$ set $\{i \in [m] : x^{\mathrm{int}}_i = 1 \}$.}
\begin{align*}
\ex{C^{(\ell)}(x^{\mathrm{int}})} &= \exc{(X_1, \dots X_m) \sim (x^{\mathrm{int}}_1, \dots x^{\mathrm{int}}_m)}{C^{(\ell)}(X_1, \dots, X_m)} \\
&\geq \exc{(X_1, \dots X_m) \sim (x^{*}_1, \dots x^{*}_m)}{C^{(\ell)}(X_1, \dots, X_m)}.
\end{align*}

To conclude it suffices to relate $\exc{(X_1, \dots X_m) \sim (x^{*}_1, \dots x^{*}_m)}{C^{(\ell)}(X_1, \dots, X_m)}$ to the value taken by the LP at the optimal solution $x^*=(x^*_1, \dots, x^*_m)$. In particular, Theorem~\ref{theorem:algorithm} directly follows from the following result (Theorem~\ref{th:main_algo}), which provides a lower bound in terms of the value achieved by the LP relaxation. 

Indeed, this randomized algorithm is quite direct: it simply solves a linear program and applies pipage rounding. We consider this as a positive aspect of the work and note that our key technical contribution lies in the underlying analysis. 

\begin{theorem}
\label{th:main_algo}
Let $x \in [0,1]^m$ and $c \in [0,1]^n$ constitute a feasible solution of the LP relaxation \eqref{lp_relaxation}. Then we have,
\begin{align*}
\exc{(X_1, \dots, X_m) \sim (x_1, \dots, x_m)}{C^{(\ell)}(X_1, \dots, X_m)} \geq \rho_{\ell} \sum_{e \in [n]} c_e
\end{align*}
where $\rho_{\ell}$ is defined by
\begin{align}
\label{eq:def-rho}
\rho_{\ell} := 1 - \frac{\ell^{\ell}}{\ell!} e^{-\ell}.
\end{align}
\end{theorem}

In fact, as we show in Lemma~\ref{lem:Ce} below, the inequality holds for every element $e \in [n]$. Before getting into the proof of this result, we establish some useful properties of the quantity $\rho_{\ell}$. 

\subsection{Preliminaries and Properties of the Approximation Ratio $\rho_{\ell}$}

Throughout, we will use $\mathrm{Poi}(\cdot)$, $\mathrm{Bin}(\cdot, \cdot)$, and $\mathrm{Ber}(\cdot)$ to, respectively, denote Poisson, Binomial, and Bernoulli random variables with appropriate parameters. 

 The next lemma uses this notion to prove the desired relation and, hence, highlights an interesting application of convex orders in the context of approximation algorithms. 
\begin{lem}
	\label{lem:convex-order-bin-poi}
	For any convex function $f$, any integer $N \geq 1$ and parameter $p \in [0,1]$, we have 
	\begin{align}
		\label{eq:cvx_order}
		\ex{f(\mathrm{Bin}(N,p))} \leq \ex{f(\mathrm{Poi}(Np))} \ .
	\end{align}
\end{lem}
\begin{proof}
	The notion of convex order between two distributions is defined as follows. If $X$ and $Y$ are random variables, we say that $X \leq_{\mathrm{cvx}} Y$ iff $\ex{f(X)} \leq \ex{f(Y)}$ holds for any convex function $f : \RR \to \RR$. We refer the reader to~\cite[Section 3.A]{bookStochOrder} for more information and properties of this order.
	As a result, the lemma will follow once we show that
	\begin{align}
		\label{eq:bin_poi_convex_order}
		\mathrm{Bin}(N,p) \leq_{\mathrm{cvx}} \mathrm{Poi}(Np).
	\end{align}
	
	First, we note that it suffices to prove this inequality for $N = 1$; this is a direct consequence of the fact that the convex order is closed under convolution~\cite[Theorem 3.A.12]{bookStochOrder} (i.e., it is closed under the addition of independent random variables).\footnote{Recall that if $X$ and $Y$ are independent, Poisson random variables with rate parameters $\lambda_1$ and $\lambda_2$, respectively, then $X+Y$ is Poisson-distributed with parameter $\lambda_1 + \lambda_2$.}
	
	Now, using~\cite[Theorem 3.A.2]{bookStochOrder}, we have that equation \eqref{eq:bin_poi_convex_order} for $N=1$ is equivalent to showing that $\ex{|\mathrm{Ber}(p) - a|} \leq \ex{|\mathrm{Poi}(p) - a|}$ for any $a \in \RR$. To prove this, we perform a case analysis. The cases $a \leq 0$ or $a \geq 1$ are simple: here, we have $\ex{|\mathrm{Ber}(p) - a|} = |p-a|$ and we always have $\ex{|\mathrm{Poi}(p) - a|} \geq \max\{\ex{\mathrm{Poi}(p) - a}, \ex{a - \mathrm{Poi}(p)}\} = |p-a|$. If $a \in (0,1)$, then 
	\begin{align*}
		\ex{|\mathrm{Ber}(p) - a|} &=  (1-p)a + p (1-a) = p+a-2ap \\
		\ex{|\mathrm{Poi}(p) - a|} &=  2a e^{-p} + p - a \geq p+a-2ap \ ,
	\end{align*}
	which concludes the proof.
\end{proof}

\label{sec:prop_rho_l}
\begin{lem}
\label{lemma:rhol-binomial}
We have
\begin{align*}
\rho_{\ell} &= 1 - \sum_{q=0}^{\ell-1} \frac{\ell - q}{\ell} \frac{\ell^{q}}{q!} e^{-\ell} =\frac{1}{\ell} \sum_{q=1}^{\ell} \pr{\mathrm{Poi}(\ell) \geq q} 
	= \frac{1}{\ell} \ex{\min\{\ell, \mathrm{Poi}(\ell)\}} \ .
\end{align*}
In addition, the following inequality holds for any $t \geq \ell$,
\begin{align}
\label{eq:rhol-binomial}
\rho_{\ell} \leq  \frac{1}{\ell} \ex{\min\left\{\ell, \mathrm{Bin}\left(t,\frac{\ell}{t}\right)\right\}} \ .
\end{align}
\end{lem}
\begin{proof}
To see the first equality, write
\begin{align*}
\sum_{q=0}^{\ell-1} (\ell - q) \frac{\ell^{q}}{q!} e^{-\ell}
&= e^{-\ell}  \sum_{m=1}^{\ell-1} \left( \frac{\ell^{q+1}}{q!} - \frac{\ell^{q}}{(q-1)!} \right) + \ell e^{-\ell} \\
&= e^{-\ell} \frac{\ell^{\ell}}{(\ell-1) !} \tag{telescoping sum}\ ,
\end{align*}
which gives the desired expression. 

The second equality is obtained by substituting the distribution function of $\mathrm{Poi}(\ell)$ and the third inequality follows from the tail-sum formula (applied over the random variable $\min\{ \ell, \mathrm{Poi}(\ell)\}$). 

To prove inequality~\eqref{eq:rhol-binomial}, it suffices to apply Lemma~\ref{lem:convex-order-bin-poi} to the concave function $\phi_{\ell} : x \mapsto \min\{x, \ell\}$ with $N = t$ and $p = \frac{\ell}{t}$.
\end{proof}

The following lemma gives a relation between the binomial distribution $\mathrm{Bin}(N,p)$ and the Poisson distribution $\mathrm{Poi}(Np)$. It is well-known that, for a constant $c$, $\mathrm{Bin}(N,c/N)$ converges to $\mathrm{Poi}(c)$ as $N$ grows. For the analysis of our algorithm, we in fact need a non-asymptotic relation between these two distributions ensuring that $\ex{\phi_{\ell}(\mathrm{Poi}(c))} \leq \ex{\phi_{\ell}(\mathrm{Bin}(N,c/N))}$, for the function $\phi_{\ell} : x \mapsto \min\{\ell, x\}$. Such a property is captured by the notion of convex order between distributions~\cite{bookStochOrder}.

 We also need a lemma about the convexity of the following function. The proof of this lemma is deferred to Appendix~\ref{sec:function-convex}.
\begin{restatable}{lem}{fnconvex} \label{lem:function-convex}
	For any nonnegative integers $s$ and $t$, the function
	\begin{align*}
		f : x \mapsto \sum_{q=0}^{s-1} (s-q) \binom{t}{q} x^{q} (1 - x)^{t-q} 
	\end{align*}
	is non-increasing and convex in the interval $[0,1]$. Note that $\binom{t}{q} = 0$ when $q > t$.
\end{restatable}

Finally, we use the following lemma used in the analysis is standard, see e.g.,~\cite{Fei09}.
\begin{lem}					
\label{lem:3-values}
	Let $x \in [0,1]^m$ be such that $x_1 + \dots + x_m = t$ with $t$ integer. We use the notation $\prc{x}{.}$ to compute probabilities where $X_1, \dots, X_m$ are independent Bernouilli random variables with $\pr{X_i=1} = x_i$. Then for any expression of the form $\sum_{\tau} a_{\tau} \prc{x}{\sum_{i=1}^m X_i = \tau}$, there exists a $q \in (0,1)$ and $x' \in [0,1]^m$ such that $x'_i \in \{0,1,q\}$ for all $i \in [m]$ and $\sum_{i} x'_i = t$ and $\sum_{\tau} a_{\tau} \prc{x}{\sum_{i=1}^m X_i = \tau} \leq \sum_{\tau} a_{\tau} \prc{x'}{\sum_{i=1}^m X_i = \tau}$. This also works for minimizing the probability. 
\end{lem}
\begin{proof}
	Let $x \in [0,1]^m$ be such that it achieves the maximum for the quantity $\sum_{\tau} a_{\tau} \prc{x}{\sum_{i=1}^m X_i = \tau}$ subject to $x_1 + \dots + x_m = t$. For the purpose of contradiction, assume $x_1$ and $x_2$ take different values in $(0,1)$. Then let $p = x_1 + x_2$. And assign $x'_1 = y$ and $x'_2 = p - y$. Then for any $\tau$, we have 
	\begin{align*}
		\prc{x'}{\sum_i X_i = \tau} &= \prc{x'}{\sum_{i\geq 3} X_i = \tau} \prc{x'}{X_1 + X_2 = 0 | \sum_{i\geq 3} X_i = \tau} \\
		&+ \prc{x'}{\sum_{i\geq 3} X_i = \tau-1} \prc{x'}{X_1 + X_2 = 1 | \sum_{i\geq 3} X_i = \tau-1} \\
		&+ \prc{x'}{\sum_{i\geq 3} X_i = \tau-2} \prc{x'}{X_1 + X_2 = 2 | \sum_{i\geq 3} X_i = \tau-2} \ .
	\end{align*}
	Consider the first term. We have $\prc{x'}{X_1 + X_2 = 0 | \sum_{i\geq 3} X_i = \tau} = (1-y)(1-(p-y))$ and so the first term is a polynomial of degree $2$ in $y$ and symmetric under the exchange $y \leftrightarrow p-y$.  The same holds for the other two terms. Furthermore, even if we are considering a sum of such terms, then we still get a symmetric polynomial of degree $2$. So the minimum and the maximum are either achieved at the boundary with $y \in \{0, p\}$ (or $\{p, 1\}$ if $p > 1$) or when $y = p - y$. This contradicts the fact that $x$ maximizes $\sum_{\tau} a_{\tau} \prc{x}{\sum_{i=1}^m X_i = \tau}$.
\end{proof}

\subsection{Proof of Theorem~\ref{th:main_algo}}

We now state and prove the main lemma for the analysis of the algorithm.
\begin{lem}
\label{lem:Ce}
Let $x \in [0,1]^m$ and $c \in [0,1]^n$ constitute a feasible solution of the linear program \eqref{lp_relaxation}. Then, we have for any $e \in [n]$:
\begin{align*}
\exc{(X_1, \dots, X_m) \sim (x_1, \dots, x_m)}{C^{(\ell)}_{e}(X_1, \dots, X_m)} \geq \rho_{\ell} c_e \ .
\end{align*}
\end{lem}
\begin{proof}
To make the notation lighter, we write $C_e^{(\ell)} = C^{(\ell)}_{e}(X_1, \dots, X_m)$ with indicators $X_i \sim \mathrm{Ber}(x_i)$. Recall that $C_e^{(\ell)} = \min\{\ell, \sum_{i \in \Gamma_e} X_i\}$, where $\Gamma_e \coloneqq \{ i \in [m] : e \in T_i \}$ denotes the indices of all the given subsets $T_i \in \mathcal{F}$  that contain the element $e$. 

The tail-sum formula gives us 
\begin{align*}
\ex{C_e^{(\ell)}} 
&= \sum_{a = 1}^{\ell} \pr{\sum_{i \in \Gamma_e} X_i \geq a} \\
&= \ell - \sum_{a=0}^{\ell-1} (\ell - a) \pr{\sum_{i \in \Gamma_e} X_i = a} \ .
\end{align*}
Now we can apply Lemma~\ref{lem:3-values} (stated and proved at the end of this section) and get that the expression for $\ex{C_e^{(\ell)}}$ is minimum when for all $i\in \Gamma_e$, $x_i \in \{0,1,q\}$ for some $q \in (0,1)$. We now assume $x$ has this form. Let $\bar{\ell}$ be the number of elements $i \in \Gamma_e$ such that $x_i = 1$. As we have in this case $\pr{ \sum_{i \in \Gamma_e} X_i = a} = 0$ for $a < \bar{\ell}$, we can write
\begin{align*}
\ex{C_e^{(\ell)}} 
&= \ell - \sum_{a=\bar{\ell}}^{\ell-1} (\ell - a) \pr{ \sum_{i \in \Gamma_e} X_i = a} \ .
\end{align*}
Note that, if $c_e \leq \bar{\ell}$, then we are done as $\ex{C^{(\ell)}_e} \geq \bar{\ell} \geq c_e$. Assume now that $c_e \geq \bar{\ell}$ and we write $d_e = c_e - \bar{\ell} \geq 0$. We also write $t$ for the number of elements $i \in \Gamma_e$ such that $x_i = q$; hence, $\sum_{i \in \Gamma_e} x_i = \bar{\ell} + qt$. 

Note that $\sum_{i \in \Gamma_e} X_i - \bar{\ell}$ has a binomial distribution with parameters $t$ and $q$. We can then write
\begin{align*}
\ex{C^{(\ell)}_e} 
&= \ell - \sum_{a=\bar{\ell}}^{\ell-1} (\ell - a) \binom{t}{a-\bar{\ell}} q^{a-\bar{\ell}} (1-q)^{t-(a-\bar{\ell})} \\
&= \ell - \sum_{a=0}^{\ell - \bar{\ell}-1} (\ell - \bar{\ell} - a) \binom{t}{a} q^{a} (1-q)^{t-a} \ ,
\end{align*}
where we implemented the change of variable $a \rightarrow a - \bar{\ell}$. We now use Lemma~\ref{lem:function-convex} with $s = \ell - \bar{\ell}$ and $t$.
Using the fact that this expression is increasing in $q$ together with the inequality $d_e \leq qt$ (this follows from the linear program~\eqref{lp_relaxation}), we get
\begin{align*}
\ex{C_e^{(\ell)}} 
&\geq \ell - \sum_{a=0}^{\ell - \bar{\ell} - 1} (\ell - \bar{\ell} - a) \binom{t}{a} \left(\frac{d_e}{t}\right)^{a} \left(1-\frac{d_e}{t}\right)^{t-a} \ .
\end{align*}
From Lemma~\ref{lem:function-convex} again, we have that the function $x \mapsto \sum_{a=0}^{\ell - \bar{\ell} - 1} (\ell - \bar{\ell} - a) \binom{t}{a} \left(\frac{x}{t}\right)^{a} \left(1-\frac{x}{t}\right)^{t-a}$ is convex in the interval $[0,t]$. We now distinguish two cases. We start with the simple case when $\ell - \bar{\ell} > t$. Then we write $d_e = \frac{d_e}{t} \cdot t + (1 - \frac{d_e}{t}) \cdot 0$ and using convexity we get
\begin{align*}
\ex{C_e^{(\ell)}} 
&= \ell - \sum_{a=0}^{t} (\ell - \bar{\ell} - a) \binom{t}{a} \left(\frac{d_e}{t}\right)^{a} \left(1-\frac{d_e}{t}\right)^{t-a} \\
&\geq \frac{d_e}{t} (\ell - (\ell - \bar{\ell} - t)) + (1 - \frac{d_e}{t}) (\ell - (\ell - \bar{\ell})) \\
&= c_e,
\end{align*}
which concludes the first case.
If $t \geq \ell - \bar{\ell}$, we instead write $d_e = \frac{d_e}{\ell - \bar{\ell}} \cdot (\ell - \bar{\ell}) + (1 - \frac{d_e}{\ell - \bar{\ell}}) \cdot 0$.
Applying convexity, we get
\begin{align*}
\ex{C^{(\ell)}_e} &\geq  \frac{d_e}{\ell - \bar{\ell}} \left(\ell - \sum_{a=0}^{\ell - \bar{\ell} - 1} (\ell - \bar{\ell} -  a) \binom{t}{a} \left(\frac{\ell - \bar{\ell}}{t}\right)^{a} \left(1-\frac{\ell - \bar{\ell}}{t}\right)^{t-a} \right) +  (1 - \frac{d_e}{\ell - \bar{\ell}}) \bar{\ell} \\
&= \bar{\ell} + \left((\ell - \bar{\ell}) - \sum_{a=0}^{\ell - \bar{\ell} - 1} (\ell - \bar{\ell} -  a) \binom{t}{a} \left(\frac{\ell - \bar{\ell}}{t}\right)^{a} \left(1-\frac{\ell - \bar{\ell}}{t}\right)^{t-a} \right) \frac{d_e}{\ell - \bar{\ell}} \\
&= \bar{\ell} + \ex{\min\{\ell - \bar{\ell}, \mathrm{Bin}(t, \frac{\ell - \bar{\ell}}{t})\}} \frac{d_e}{\ell - \bar{\ell}} \\
&\geq \bar{\ell} + \rho_{\ell - \bar{\ell}} \cdot d_e \\
&= \left(\frac{\bar{\ell}}{c_e} + \rho_{\ell - \bar{\ell}} \frac{c_e - \bar{\ell}}{c_e} \right) c_e \ ,
\end{align*}
where to obtain the last inequality, we used equation~\eqref{eq:rhol-binomial} from Lemma~\ref{lemma:rhol-binomial}.
Now observe that $\left(\frac{\bar{\ell}}{c_e} + \rho_{\ell - \bar{\ell}} \frac{c_e - \bar{\ell}}{c_e} \right) = \frac{\bar{\ell}(1-\rho_{\ell - \bar{\ell}})}{c_e} + \rho_{\ell - \bar{\ell}}$ is a decreasing function of $c_e$ and so we can lower bound it with the value it takes when $c_e = \ell$. So it only remains to show that 
\begin{align*}
\frac{\bar{\ell}}{\ell} + \rho_{\ell - \bar{\ell}} \frac{\ell - \bar{\ell}}{\ell}  \geq \rho_{\ell} \ .
\end{align*}
Recalling that $\rho_{\ell} = 1 - \frac{\ell^{\ell} e^{-\ell}}{\ell!}$, this is equivalent to 
\begin{align*}
\frac{(\ell-\bar{\ell})^{\ell-\bar{\ell}} e^{-(\ell-\bar{\ell})}}{(\ell-\bar{\ell})!} \frac{\ell - \bar{\ell}}{\ell}  \leq  \frac{\ell^{\ell} e^{-\ell}}{\ell!} \ .
\end{align*}
In other words, it suffices to show that the sequence $\frac{\ell^{\ell+1} e^{-\ell}}{\ell!}$ is an increasing sequence. To see this, we can take the logarithm of the ratio of the $\ell$th term to the $(\ell+1)$th term and get $(\ell+1) \ln(1-\frac{1}{\ell+1}) + 1 \leq 0$.
%
\end{proof}

\subsection{Generalization to weighted cover subject to a matroid constraint}
\label{sec:generalization}
As mentioned previously, the algorithm can be easily generalized by allowing the objective function to be a weighted $\ell$-coverage function and the constraint to be one that requires $S \in \mathcal{M}$, for a matroid $\mathcal{M}$. More precisely, we are now given a collection of real weights $\{w_{i,e}\}_{i \in [m], e \in [n]}$. For an integer $\ell$ and a set $S \subseteq [m]$, we define the weighted $\ell$-coverage of an element $e$ to be $C_e^{(\ell)}(S) \coloneqq  \max_{i_1, \dots, i_{\ell} \in S } w_{i_1, e} + \dots + w_{i_{\ell}, e}$;  here the maximization is over all distinct indices $i_1, \ldots, i_{\ell}$ in the set $S$. In other words, $C_{e}^{(\ell)}(S)$ is the sum of the largest $\ell$ weights in the list $(w_{i,e})_{i \in S}$. Then, as before, we define $C^{(\ell)}(S) = \sum_{e \in [n]} C_{e}^{(\ell)}(S)$. The problem at hand is to maximize $C^{(\ell)}(S)$ subject to the matroid constraint $S \in \mathcal{M}$.

The algorithm has exactly the same structure as the one described at the beginning of this section. We consider the following linear program.
\begin{align}\label{lp_relaxation_w}
\begin{aligned}
& \underset{x, c}{\max}
& & \sum_{i \in [m], e \in [n]} c_{i,e} w_{i,e} \\
& \text{subject to}
& &\sum_{i \in [m]} c_{i,e} \leq \ell \quad \forall e \in [n] \\
& 
& &c_{i,e} \leq x_i  \quad \forall i \in [m], e \in [n] \\
& 
& & 0 \leq x_i \leq 1 \quad \forall i \in [m] \\ 
&
& & x \in P(\mathcal{M}) \,.
\end{aligned}
\end{align}
Here $P(\mathcal{M})$ is the matroid polytope of $\mathcal{M}$. Note that, as before, the function $C^{(\ell)} (\cdot)$ is submodular and~\cite{calinescu2011maximizing} gives an efficient algorithm for computing a $1-e^{-1}$ approximation for such problems and more generally sums of weighted rank functions with matroid constraints.
Again, once we obtain an optimal fractional solution $x^*$, we use pipage rounding to obtain a random integral vector $x^{\mathrm{int}} \in \{0,1\}^m$, such that $x^{\mathrm{int}} \in \mathcal{M}$ and $\ex{C^{(\ell)}(x^{\mathrm{int}})} \geq \exc{(X_1, \dots X_m) \sim (x^{*}_1, \dots x^{*}_m)}{C^{(\ell)}(X_1, \dots, X_m)}$ (see, e.g., \cite[Lemma 3.5]{calinescu2011maximizing}). Thus, it only remains to relate this expectation to the objective value achieved by the linear program.
For a fixed $e \in [n]$, we can express the weighted coverage function $C^{(\ell)}_e(\cdot)$ as follows. First order the weights so that 
$w_{1,e} \geq w_{2, e} \geq \dots \geq w_{m,e}$. Then, for $X_1, X_2, \ldots, X_m \in \{0,1\}$, write $C_{i, e} = \min\{\ell, \sum_{j=1}^{i} X_j\} - \min\{\ell, \sum_{j=1}^{i-1} X_j\}$ for all $i \in [m]$. Note that $C^{(\ell)}_e(X_1, \dots, X_m) = \sum_{i \in [m]} w_{i,e} C_{i,e}$. Thus, the following lemma is sufficient to obtain the desired result.
\begin{lem}
Let $e \in [n]$. Assume that $w_{1,e} \geq w_{2, e} \geq \dots \geq w_{m,e}$. Let $x, c$ be a feasible solution of the above linear program and $X_1, X_2, \ldots, X_m \in \{ 0,1 \}$ be independent random variables with $\pr{X_i = 1} = x_i$. Then, for 
$C_{i, e} = \min\{\ell, \sum_{j=1}^{i} X_j\} - \min\{\ell, \sum_{j=1}^{i-1} X_j\}$, we have
\begin{align*}
\sum_{i \in [m]} w_{i,e} \ex{C_{i,e}} \geq \rho_{\ell} \sum_{i \in [m]} w_{i,e} c_{i,e}.
\end{align*}
\end{lem}
\begin{proof}
Fixing $w_{m+1,e} = 0$ and using the definition of $C_{i,e}$, we can write
\begin{align*}
\sum_{i=1}^m w_{i,e} \ex{C_{i,e}} &= \sum_{i=1}^m w_{i,e} \ex{\min\{ \ell, \sum_{j=1}^{i} X_j \} - \min\{ \ell, \sum_{j=1}^{i-1} X_j \} } \\
&= \sum_{i=1}^m (w_{i,e} - w_{i+1,e}) \ex{\min\{ \ell, \sum_{j=1}^{i} X_j \}} \,.
\end{align*}
Our objective now is to show that for every $i$, 
\begin{align}
\label{eq:weighted_version}
\ex{\min\{ \ell, \sum_{j=1}^{i} X_j \}} \geq \rho_{\ell} \sum_{j=1}^{i} c_{j,e} \ .
\end{align}
But this follows from Lemma~\ref{lem:Ce} applied to the setting where we only consider sets $j \in \{1, \dots, i\}$ and replacing $\sum_{j=1}^i c_{j,e}$ with $c_e$. Thus we get 
\begin{align*}
\sum_{i=1}^m w_{i,e} \ex{C_{i,e}} &\geq \rho_{\ell} \sum_{i=1}^m (w_{i,e} - w_{i+1,e}) \sum_{j=1}^{i} c_{j,e} = \rho_{\ell} \sum_{i=1}^{m} w_{i,e} c_{i,e} \ .
\end{align*}
This concludes the proof of the lemma.
\end{proof}

\input{general-hardness-of-approx}

\section{Concluding remarks} 

The standard coverage function $C(S)$ counts the number of elements $e \in [n]$ that are covered by at least one set $T_i$ with $i \in S$. Note that the contribution of an $e \in [n]$ to $C(S)$ is exactly the same whether $e$ appears in just one set $T_i$ or in all of them. As previously mentioned, it is very natural to consider settings wherein having more than one copy of $e$ is more valuable than just one copy of it. The $\ell$-coverage function we introduced does exactly that: having $c$ copies of element $e$ has a value of $\min\{\ell, c\}$. We showed that when this is the case, we can take advantage of this structure and obtain a better approximation guarantee as a function of $\ell$. In subsequent works~\cite{dudycz2020tight,barman2020tight}, the general setting where $c$ copies have a value of $\varphi(c)$ for some concave and monotone function $\varphi$ was studied.
More generally, we believe that our work paves the way towards an operationally motivated notion of submodularity that interpolates between linear functions and completely general submodular functions. The previously mentioned notion of curvature studied in~\cite{CC84,SVW17} does this interpolation but the definition is unfortunately too restrictive and thus difficult to interpret operationally. 

Another interesting question is whether there exists combinatorial algorithms that achieve the approximation ratio $\rho_{\ell}$ for maximum $\ell$-coverage for $\ell \geq 2$. For $\ell =1$, the simple greedy algorithm does the job, but as we mentioned previously, the greedy algorithm only gives a $1-e^{-1}$ approximation ratio even for $\ell \geq 2$. Is it possible to generalize the greedy algorithm to give an approximation ratio beating $1-e^{-1}$?

\section*{Acknowledgements}

We are very grateful to Guillaume Aubrun for referring us to the notion of convex order between distributions and  Bar{\i}\c{s} Nakibu\u{g}lu for asking us about the list-decoding variant of~\cite{BF15} during the workshop Beyond IID in Information Theory held at the Institute for Mathematical Sciences, National University of Singapore in 2017. We would also like to thank Edouard Bonnet for referring us to the literature on set multicover. In addition, we thank the reviewers for their very useful comments on the manuscript, in particular for pointing out an error in the proof of Lemma~\ref{lem:gadget} in a previous version and suggesting to use negative association to correct it.

This research is supported by the French ANR project ANR-18-CE47-0011 (ACOM). Siddharth Barman gratefully acknowledges the support of a Ramanujan Fellowship (SERB - {SB/S2/RJN-128/2015}) and a Pratiksha Trust Young Investigator Award. Part of this work was conducted during the first author's visit to \'{E}cole Normale Sup\'{e}rieure de Lyon and was supported by the Administration de la recherche (ADRE), France.

  \bibliographystyle{abbrv}
  
  \bibliography{big}
  
  \appendix

  \input{convexity-lemma}
  \input{concavity}

\input{hypergraph-ugc}

\end{document}

%% file: general-hardness-of-approx.tex
\allowdisplaybreaks
\section{Hardness of Approximating the Multi-Coverage Problem}

In this section we establish an inapproximability bound for the maximum $\ell$-multi-coverage problem. Throughout this section we will use $\Gamma$ to denote the universe of elements and, hence, an instance of the $\ell$-multi-coverage problem will consist of $\Gamma$, along with a collection of subsets $\mathcal{F} = \{T_i \subseteq \Gamma\}_{i=1}^m$ and an integer $k$. Recall that the  objective of this problem is to find a size-$k$ subset $S \subseteq [m]$ that maximizes $C^{(\ell)}(S) = \sum_{e \in \Gamma} \min\{\ell, |\{i \in S : e \in T_i\}| \}$. 

Formally, we establish Theorem~\ref{theorem:hardness} (restated below for completeness).\footnote{As mentioned above, for brevity, we will sometimes refer to $\ell$-multi-coverage as $\ell$-coverage.}  

\hardness*

Our reduction is from the $h$-$\kary$ problem, as detailed in Definition~\ref{def:unique}. Specifically, we will rely on Conjecture~\ref{conj:ugc}, which asserts the hardness of  the $h$-$\kary$ problem by showing that it is equivalent to the Unique Games Conjecture ({\rm UGC}). 


\begin{definition}[$h$-$\kary$]							\label{def:unique}
	An instance  $\mathcal{G}(V,E,\Sigma, \{\pi_{e,v}\}_{e\in E,v \in e})$ of $h$-$\kary$ is characterized by an $h$-uniform regular hypergraph, $(V, E)$, and bijection (projection) constraints $\pi_{e,v}: \Sigma \mapsto \Sigma$. Here, each $h$-uniform hyperedge represents a $h$-ary constraint. Additionally, for any labeling $\sigma : V \mapsto \Sigma$, we have the following notions of strongly and weakly satisfied constraints:
	\begin{itemize}
		\item A hyperedge $e = (v_1,v_2,\ldots,v_h) \in E$ is \emph{strongly satisfied} by $\sigma$ if for every $x,y \in [h]$ we have $\pi_{e,v_x}(\sigma(v_x)) = \pi_{e,v_y}(\sigma(v_y))$.
		\item A hyperedge $e = (v_1,v_2,\ldots,v_h) \in E$ is \emph{weakly satisfied} by $\sigma$ if there exists $x,y \in [h], x \neq y$ such that $\pi_{e,v_x}(\sigma(v_x)) = \pi_{e,v_y}(\sigma(v_y))$.
	\end{itemize}
\end{definition}

The following conjecture is equivalent to {\rm UGC} (see, Appendix \ref{sec:hyper-ugc}): 

\begin{conjecture}
	\label{conj:ugc}
	For any constant $\varepsilon > 0$ and constant integer $h \geq 2$, given an instance $\mathcal{G}$ of $h$-$\kary$, it is {\rm NP}-hard to distinguish between 
	\begin{itemize}
		\item(YES): There exists a labeling $\sigma$ that strongly satisfies at least $1 - \varepsilon$ fraction of the edges.
		\item(NO): No labeling weakly satisfies more than $\varepsilon$ fraction of the edges.
	\end{itemize}
	Furthermore, the constraint hypergraph for instance $\mathcal{G}$  is \emph{regular} and the size of the underlying alphabet set depends only on the parameter $\varepsilon$.	 
\end{conjecture}

Our reduction uses a gadget that generalizes a partitioning system used by Feige~\cite{Fei98}. We will begin by describing the partitioning gadget (Section~\ref{sect:ell-part-gadget}) and then the use it for the reduction (Section~\ref{sect:ell-reduction}).

\subsection{The Partitioning Gadget}
\label{sect:ell-part-gadget}

For any set of elements $[\widehat{n}]$ (with $\widehat{n} \in \mathbbm{Z}_+$) and a collection of subsets $\mathcal{Q} \subseteq 2^{[\widehat{n}]}$, we use $C^{(\ell)}(\mathcal{Q})$ to denote the $\ell$-coverage of $[\widehat{n}]$ by the subsets contained in  $\mathcal{Q}$,\footnote{Here, the notation $C^{(\ell)}(\cdot)$ is overloaded for ease of presentation.} i.e., $C^{(\ell)}(\mathcal{Q}) := \sum_{e \in [\widehat{n}]} \ \min  \Big\{ \ell,  \  | \{ P \in \mathcal{Q} : e \in P\}| \ \Big\}$. Furthermore, we will say that the collection of subsets $\mathcal{Q}$ is an $\ell$-cover of $S$ if $C^{(\ell)}(\mathcal{Q}) = \ell \widehat{n}$, i.e., if every element of $[\widehat{n}]$ is covered at least $\ell$ times.

We begin by defining the $([\widehat{n}], h, s, \ell, \eta)$-partitioning system which is the basic gadget used in our reduction. 

\begin{definition}					\label{def:ell-partition}
	Given a ground set $[\widehat{n}]$, an $([\widehat{n}],h, s, \ell, \eta)$-partitioning system consists of $s$ collections of subsets of $[\widehat{n}]$ denoted $\mathcal{P}_1, \mathcal{P}_2, \ldots, \mathcal{P}_s$ that satisfy  
	
	\begin{itemize}
		\item[(1)] For every $i \in [s]$, the family $\mathcal{P}_i$ is a collection of $h$ subsets $P_{i,1}, P_{i,2},\ldots, P_{i,h} \subseteq [\widehat{n}]$ such that: (i) for every $j \in [h]$, $|P_{i,j}| = \ell \widehat{n}/h$ and (ii) for every $a \in [\hat{n}]$, $|\{i : a \in P_{i,j}\}| = \ell$. 
		
		In other words, $\mathcal{P}_i$ forms an $\ell$-cover of the ground set $[\widehat{n}]$ 
		\item[(2)] For any subset of indices $T \subseteq [s]$ and any collection of subsets $\mathcal{Q} = \{P_{i,j(i)} \mid i \in T \}$ for some function $j : T \to [h]$, we have
		$C^{(\ell)}(\mathcal{Q}) \le \left( \fcov^{(\ell)}_{|T|, h} + \eta \right) \widehat{n}$, where 
		\begin{equation*}
		\fcov^{(\ell)}_{|T|, h} := \ell - \sum_{i = 0}^{\ell-1}(\ell - i){ |T| \choose i } \left( \frac{\ell}{h} \right)^i  \left(1- \frac{\ell}{h}\right)^{|T|-i}.
		\end{equation*}
		In particular, if $|T| = h (1 + \mu)$, for some $0 < \mu < \frac{\varepsilon}{ 2\ell^2 } \ e^\ell  \frac{\ell!}{\ell^\ell} $, then
		\begin{align}								\label{eq:rho-ratio}
		C^{(\ell)} (\mathcal{Q}) \le \left( \ell \left(1 -  \frac{\ell^{\ell} }{ \ell!} \ e^{-\ell}  \right) + \varepsilon \right) \widehat{n}. 
		\end{align}
		
		\end{itemize}
\end{definition}

Note that we can restate equation \eqref{eq:rho-ratio} as $C^{(\ell)} (\mathcal{Q}) \le \left( \rho_\ell\ell + \varepsilon \right) \widehat{n}$. We now state and prove a lemma which shows that such partitioning systems exist for a useful range of parameters.

\begin{lemma}			\label{lem:gadget}
	For every choice of $\widehat{n}, h, s, \ell \in \mathbbm{Z}_+$ and $\eta > 0$ such that $s \geq h \geq \ell$ and ${\widehat{n}} \ge 100\eta^{-2}s \ell^2 \log h$ a multiple of $h$,  there exists an $([\widehat{n}], h, s, \ell,\eta)$-partitioning system. Moreover, such a partitioning system can be found in time $\exp(s \widehat{n} \log \widehat{n})\cdot{\rm poly}(h)$. 
\end{lemma}

\begin{proof}
The existential proof is based on the probabilistic method. Given integers $\widehat{n}, s , h$, for every $i \in [s]$ we set the collection of subsets $\mathcal{P}_i = (P_{i,1}, \dots, P_{i, h})$ to be uniformly chosen among all the collections satisfying condition (1), i.e., the sets are all of size $\frac{\ell \widehat{n}}{h}$ and they form an $\ell$-cover of $[\widehat{n}]$. This can be achieved by first choosing the set $P_{i,1}$ uniformly at random among all subsets of $[\widehat{n}]$ of size $\frac{\ell \widehat{n}}{h}$, then $P_{i,j}$ is chosen uniformly at random among all size-$\frac{\ell \widehat{n}}{h}$ subsets of $[\widehat{n}]$ from which we remove the elements that already appeared $\ell$ times in $P_{i,1}, \dots, P_{i,j-1}$.
Write  $\mathcal{P} = (\mathcal{P}_1, \mathcal{P}_2,\ldots, \mathcal{P}_s)$. Note that for each element $a \in [\widehat{n}]$, and any subset $P_{i, j} \in \mathcal{P}_i$, we have $\Pr \ [ a \in P_{i,j} ] = \ell/h$. By construction, condition $(1)$ of Definition~\ref{def:ell-partition} is satisfied, hence all that remains is to show that the partitioning system $\mathcal{P}$ also satisfies condition $(2)$. 
	
	Towards that, we fix an index set $T \subseteq [s]$ and a collection of subsets $\mathcal{Q} = \{P_{i, j(i)} | i \in T\}$. For ease of notation, we will denote $P_{i, j(i)}$ by $P_i$; hence, $\mathcal{Q} = \{P_i\}_{i \in T}$. For a fixed element $a \in [\widehat{n}]$, the expected $\ell$-coverage of $\mathcal{Q}$ is equal to
	\begin{align*}
	&\sum_{i = 1}^{\ell-1} i \cdot \Pr \Big[ \left| \{ P \in \mathcal{Q} : P \ni a  \} \right| = i \Big] + \ell \cdot \Pr\Big[  \left| \{ P \in \mathcal{Q} : P \ni a  \} \right| \geq \ell \Big] \\
	& = \sum_{i = 1}^{\ell-1} i \cdot \Pr \Big[ \left| \{ P \in \mathcal{Q} : P \ni a  \} \right| = i \Big] + \ell \left(1-  \sum_{i = 0}^{\ell-1} \Pr \Big[ \left| \{ P \in \mathcal{Q} : P \ni a  \} \right| = i \Big] \right) \\
	& = \ell - \sum_{i = 0}^{\ell-1} (\ell-i)  \Pr \Big[ \left| \{ P \in \mathcal{Q} : P \ni a  \} \right| = i \Big]  \\
	& = \ell - \sum_{i = 0}^{\ell-1} (\ell-i){|T| \choose i} \left(\frac{\ell}{h} \right)^i \left(1 - \frac{\ell}{h} \right)^{|T| - i}  \tag{$\mathcal{P}_is$ are constructed independently},
	\end{align*}
	where in the last line, we used the fact that all the sets in $\mathcal{Q}$ are independent and uniformly chosen subsets of $[\widehat{n}]$ of size $h$.
	We now denote by $X^a_{i} \in \{0,1\}$ the indicator of the event $[ a \in P_i ]$. Then, $C^{(\ell)} ( \mathcal{Q}) = \sum_{a \in [\widehat{n}]} C_{a} $ with $C_a = \min\{\ell, \sum_{i=1}^h  X^{a}_i\}$. We just computed $\ex{C_a}$ so by linearity of expectation, we have $\E_{\mathcal{P}} \left[ C^{(\ell)} ( \mathcal{Q}) \right] = \fcov^{(\ell)}_{|T|,h} \ \widehat{n}$. Now, we claim that  
	\begin{equation}					\label{eqn:conc1}
	\Pr_\mathcal{P} \left[ C^{(\ell)} (\mathcal{Q})  > \left( \fcov^{(\ell)}_{|T|,h} + \eta \right) \widehat{n} \right] \le 2 \exp\Big(-2 (\eta/\ell)^2 \ {\widehat{n}} \Big) \ .
	\end{equation}
	In order to establish this result, we first observe that $0 \leq C_a \leq \ell$, but the random variables $C_a$ are not independent in general. However, the random variables $\{C_a\}_{a \in [\widehat{n}]}$ are negatively associated~\cite{joag1983negative} and this is sufficient for the Chernoff-Hoeffding bound; see e.g.,~\cite[Chapter 3]{dubhashi2009concentration} or~\cite{dubhashi1996balls}. The set of random variables $\{C_a\}_{a \in [\widehat{n}]}$ is said to be negatively associated if for any functions $f$ and $g$ either both increasing or both decreasing and any disjoint index sets $I, J \subseteq [\widehat{n}]$, we have
	\begin{align*}
	\ex{f(C_a : a \in I) \cdot g(C_a : a \in J)} \leq \ex{f(C_a : a \in I)} \cdot \ex{g(C_a : a \in J)} \ .
	\end{align*}
	
	Note that the random variable $C_a$ is a monotone increasing function of $\{X_i^{a}\}_{i \in [s]}$. Thus, in order to show that $\{C_a\}_{a \in [\widehat{n}]}$ are negatively associated, it suffices to show that $\{X_i^{a}\}_{i \in [s], a \in [\widehat{n}]}$ are negatively associated.
	
	For this, observe that for any fixed $i \in [s]$, $\{X_i^{a}\}_{a \in [\widehat{n}]}$ are negatively associated because it corresponds to a permutation distribution (see e.g.,~\cite[Theorem 2.11]{joag1983negative}). Then, using the fact that the families $\{X_i^{a}\}_{a \in [\widehat{n}]}$ for $i \in [s]$ are mutually independent, (see e.g.,~\cite[Property P7]{joag1983negative}) we obtain that $\{X_i^{a}\}_{i \in [s], a \in [\widehat{n}]}$ are negatively associated. We can thus apply the Chernoff-Hoeffding bound for negatively associated random variables (see e.g.,~\cite[Chapter 3]{dubhashi2009concentration} or~\cite{dubhashi1996balls}) and get~\eqref{eqn:conc1}. 
	
	
	Note that equation (\ref{eqn:conc1}) holds for a fixed choice of $T$ and $\mathcal{Q}$. Applying union bound over all the $(h+1)^s$ possible choices of $T$ and $\mathcal{Q}$, we have that with probability at least $0.9$, the $\ell$-covering value  satisfies $C^{(\ell)}(\mathcal{Q}) \le \left( \fcov^{(\ell)}_{|T|,h}  +  \eta \right) \widehat{n} \ $ (since ${\widehat{n}} \ge 100 s \ell^2 \eta^{-2}\log h$).

	To prove the last statement of the lemma, we consider $|T| = h(1 + \mu)$ for $0 < \mu < \frac{\varepsilon}{ \ell } \  \frac{\ell!}{\ell^\ell} $. Then, it follows that for a large enough (compared to $\ell$) $h$ and $i \leq |T|$, we have $\left(1 - \frac{\ell}{h} \right)^{|T| - i} \ge \left(1 - \frac{\ell}{h}\right)^{|T|} \ge e^{-\ell (1 + 2\mu)}$. Therefore,  $\fcov^{(\ell)}_{|T|,h}$ can be upper bounded as 
	\begin{align}			
	\ell - \sum_{i = 0}^{\ell-1} (\ell-i){|T| \choose i} \left(\frac{\ell}{h} \right)^i \left(1 - \frac{\ell}{h} \right)^{|T| - i}  &  \leq  \ell - e^{-\ell (1 + 2\mu)}\sum_{i = 0}^{\ell-1} (\ell - i ) {|T| \choose i} \left( \frac{\ell}{h} \right)^i   \label{eq:sum-term}
	\end{align}
	
	For $(\ell -1) = o(\sqrt{|T|})$ and $i \leq \ell-1$, the following bound holds for the binomial coefficients 
\begin{align}						\label{eq:tchoosei}
{|T| \choose i} &  \geq ( 1 - o(1)) \frac{|T|^{i}}{i!} \\ 
& = ( 1 - o(1)) \frac{h^i (1+\mu)^{i}}{i!} \ .
\end{align}


Therefore, the right-hand-side of inequality \eqref{eq:sum-term} satisfies{\footnote{Here we ignore the $(1 - o(1))$ multiplicative factor from equation \eqref{eq:tchoosei} to keep the calculation clean. The $(1 - o(1))$ term can be accounted for in equation \eqref{eq:last-step}, where it can be absorbed into $\varepsilon$.}} 
\begin{align}
 \ell - e^{-\ell (1 + 2\mu)}\sum_{i = 0}^{\ell-1} (\ell - i ) {|T| \choose i} \left( \frac{\ell}{h} \right)^i & \leq \ell - e^{-\ell (1 + 2\mu)} \sum_{i=0}^{\ell-1} \left( \ell - i \right) \frac{h^i (1+\mu)^{i}}{i!} \left( \frac{\ell}{h} \right)^i  \notag \\
 & \leq \ell - e^{-\ell (1 + 2\mu)} \sum_{i=0}^{\ell-1} \left( \ell - i \right) \frac{\ell^i}{i!}  \tag{since $(1+\mu)^i \geq 1$} \\ 
 & = \ell - e^{-\ell (1 + 2\mu)} \left( \sum_{i=1}^{\ell-1}  \left( \frac{\ell^{i+1}}{i!} - \frac{\ell^i}{(i-1)!} \right) \ + \ell \right) \notag \\
 & = \ell - e^{-\ell (1 + 2\mu)} \frac{\ell^\ell}{(\ell-1)!} \tag{telescoping sum} \\
 & \leq \ell - e^{-\ell} (1 - 2\ell \mu)  \frac{\ell^\ell}{(\ell-1)!} \tag{$e^{-x} \geq 1 - x$}  \\
 & = \ell - e^{-\ell} \frac{\ell^\ell}{(\ell-1)!} + e^{-\ell} \ell \frac{\ell^\ell}{(\ell-1)!}\cdot 2\mu \notag \\
 & \leq \ell - e^{-\ell} \frac{\ell^\ell}{(\ell-1)!} + \varepsilon	\ .					\label{eq:last-step}
 \end{align}
 Here, the last inequality follows from the fact $\mu < \frac{\varepsilon}{ 2\ell^2 } \ e^\ell  \frac{\ell!}{\ell^\ell} $. Overall, in this setting we get that 
 \begin{align*}
\fcov^{(\ell)}_{|T|,h} & \leq \ell \left( 1 -  \frac{\ell^\ell}{\ell!} e^{-\ell} \right) + \varepsilon \ .
 \end{align*}
Therefore, $C^{(\ell)}(\mathcal{Q})$ satisfies the stated bound $C^{(\ell)} (\mathcal{Q}) \le \left( \ell \left(1 -  \frac{\ell^{\ell} }{ \ell!} \ e^{-\ell}  \right) + \varepsilon \right) \widehat{n}$.
	
	
	
	
	Since, a random choice of $\mathcal{P}=(\mathcal{P}_1, \mathcal{P}_2, \ldots, \mathcal{P}_s)$ satisfies the desired properties, we can enumerate over all choices of $(\mathcal{P}_1, \mathcal{P}_2, \ldots, \mathcal{P}_s)$ in time $\exp(s \widehat{n} \log \widehat{n}){\rm poly}(h)$ to find such a partitioning system.\footnote{Note that in our setting, $m$, $h$, and $s$ will be treated as constants and, hence, such a partitioning system can be constructed in constant time.}
	\end{proof}

Finally, we will also need the notion of \emph{piecewise linear extension} of a function defined on integers.

\begin{definition}[Piecewise Linear Extension]
	Let $f:\mathbbm{Z}_+ \mapsto \mathbbm{R}$ be a function defined on the nonnegative integers. We denote its piecewise linear extension $\tilde{f}:\mathbbm{R}_+ \mapsto \mathbbm{R}$ as $\tilde{f}(a) \coloneqq \lambda f(i) + (1 - \lambda)f(i+1)$, where $a \in \mathbb{R}_+$ lies between the integers $i$ and $i+1$ (i.e., $a \in [i,i+1)$) and $\lambda$ satisfies $a = \lambda i + (1-\lambda)(i+1)$.
\end{definition}

By definition, $\tilde{f}(a) = f(a)$ for every integer $a \in \mathbbm{Z}_+$. Consequently, for any distribution $\mathcal{D}$ supported over the integers, we have $\E_{X \sim \mathcal{D}}\left[f(X)\right] = \E_{X \sim \mathcal{D}}\left[\tilde{f}(X)\right]$.

\subsection{The Reduction}
\label{sect:ell-reduction}

We now describe the reduction from $h$-$\kary$ to the multi-coverage problem. Given an instance $\mathcal{G}(V,E,\Sigma, \{\pi_{e,v}\}_{e\in E,v \in e})$ of $h$-$\kary$ (as described in Definition \ref{def:unique}), we construct an instance of the maximum $\ell$-multi-coverage problem with ground set ${\Gamma}$ and set system $\mathcal{F}$ as follows.

	\noindent {\bf Ground Set}: For every $h$-uniform hyperedge $e \in E$, we introduce a distinct copy of the set $[\widehat{n}]$, which we denote by $[\widehat{n}]_e$. The overall ground set is $\Gamma = \cup_{e \in E} [\widehat{n}]_e$.

	\noindent {\bf Set System}: Fix an $([\widehat{n}], h, s, \ell,\eta)$-partitioning gadget as described in Lemma \ref{lem:gadget}, with $s = |\Sigma|$ i.e., the size of the alphabet set from the $h$-$\kary$ instance. For each hyperedge $e = (v_1,v_2,\ldots,v_h) \in E$, consider the copy of the $([\widehat{n}], h, s, \ell,\eta)$-partitioning gadget on the elements corresponding to the hyperedge $e$. Say for $e$, the partitioning gadget is comprised of the collections $\mathcal{P}^e_1, \mathcal{P}^e_2, \ldots, \mathcal{P}^e_s$. 
	
	Using these collections, we will first define sets $T^{e,v}_\beta$ for each hyperedge $e \in E$, vertex $v \in e$, and alphabet $\beta \in \Sigma$ in the given instance $\mathcal{G}$. Then, for every vertex $v$ and alphabet $\beta$, we will include $\cup_{e \in E : e \ni v} \  T^{e,v}_\beta$ as a subset in the set system $\mathcal{F}$. 
	
	 For each $i \in [s]$, we consider the $i$th alphabet of $\Sigma$, say $\alpha_i$, and associate the $h$ sets in $\mathcal{P}^e_i$ with labels for $v_1,v_2, \ldots,v_h$ which map (under bijection $\pi_{e, v_j}$) to $\alpha_i$. This is done by renaming the subsets in $\mathcal{P}^e_i = \{P^e_{i,1}, \ldots, P^e_{i,h} \}$ to $T^{e,v_1}_{(\pi_{e,v_1})^{-1}(\alpha_i)},T^{e,v_2}_{(\pi_{e,v_2})^{-1}(\alpha_i)},\ldots,T^{e,v_h}_{(\pi_{e,v_h})^{-1}(\alpha_i)}$, respectively. In other words, if alphabet  $\beta \in \Sigma$ satisfies $\pi_{e, v_j} (\beta) = \alpha_i$, then we assign $T^{e, v_j}_{\beta} = P^e_{i, j}$.  Since $\pi_{e, v_j}$ is a bijection, its inverse (i.e., $\beta$) is well-defined.  
	
	For every vertex $v$ and alphabet $\beta \in \Sigma$,  we construct $\widetilde{T}^v_\beta := \cup_{e \in E} T^{e,v}_{\beta}$. These subsets $\widetilde{T}^v_\beta \subseteq \Gamma = \cup_{e \in E} [\widehat{n}]_e$ constitute the set system of the $\ell$-coverage problem, i.e., $\mathcal{F} := \Big\{\widetilde{T}^v_\beta \mid  v \in V, \beta \in \Sigma \Big\}$. Finally, we set the (cardinality constraint) threshold $k = |V|$, i.e., in the constructed instance $(\Gamma, \mathcal{F})$ the objective is to select $k =|V|$ subsets (from $\mathcal{F}$) with as large an $\ell$-coverage value as possible.  
	
	We quickly point out that in order to ensure that guarantees of the partitioning gadgets hold, we need $|\Sigma| = s \ge |T| \ge 2h$. However, since $h$ is a constant, without loss of generality, we can always consider $h$-$\kary$ instances with alphabet sizes large enough without losing out on the completeness and soundness parameters. We detail this observation in Section \ref{sec:alphabet}.

\subsection{Proof of Theorem~\ref{theorem:hardness}}	
	In the following subsections we argue the completeness and soundness directions of the reduction. This will establish the stated inapproximability result. 
		
\subsubsection{Completeness}

Suppose the given $h$-$\kary$ instance $\mathcal{G}$ is a YES instance. Then, there exists a labeling $\sigma: V \mapsto \Sigma$ which strongly satisfies $1 - \delta$ fraction of edges. Consider the collection of $|V|$ subsets $\mathcal{T} := \{ \widetilde{T}^v_{\sigma(v)} \mid v \in V \}$. Let $e = ({v_1,v_2,\ldots,v_h})$ be a hyperedge which is strongly satisfied by $\sigma$. Then, $\pi_{e,v_1}(\sigma(v_1)) = \pi_{e,v_2}(\sigma(v_2))=\cdots=\pi_{e,v_h}(\sigma(v_h)) = \alpha_r$, for some $\alpha_r \in \Sigma$. By construction, for every $x \in [h]$, we have $T^{e, v_x}_{\sigma(v_x)} \subset \widetilde{T}^{v_x}_{\sigma(v_x)}$ and $\mathcal{P}^e_{r} = \left\{ T^{e,v_x}_{\sigma(v_x)} \mid x \in [h] \right\}$. Therefore, condition $(1)$ of Definition \ref{def:ell-partition} ensures that the sets $\widetilde{T}^{v_x}_{\sigma(v_x)}$ forms an $\ell$-cover of $[\widehat{n}]_e$. Since this is true for at least $(1 - \delta)$ fraction of the edges, we have $C^{(\ell)}(\mathcal{T}) \geq (1 - \delta) \ell |\Gamma|$.

\subsubsection{Soundness}

	We establish the contrapositive for the soundness claim, i.e., if there exists a family of $|V|$ sets with large $C^{(\ell)}$-value, then there exists a labeling which \emph{weakly} satisfies a significant fraction of the edges of the given instance $\mathcal{G}$. Formally, suppose there exists a collection $\mathcal{T} \subset \mathcal{F}$ of $|V|$ subsets such that 
	\begin{align}
	C^{(\ell)}( \mathcal{T}) & \geq \left(  1-  \frac{\ell^\ell}{\ell!} e^{-\ell}  \right) \ell |\Gamma| + \delta |\Gamma| \label{ineq:contrap}
	\end{align}
	For every vertex $v \in V$, we define $L(v) := \{\beta \in \Sigma \mid \widetilde{T}^v_{\beta} \in \mathcal{T}\}$ to be the candidate set of labels that can be associated with the vertex $v$. We extend this definition to hyperedges $e = (v_1,v_2,\ldots,v_h)$, where we define $L(e) := L(v_1) \cup L(v_2) \cup \cdots L(v_h)$ to be the \emph{multiset} of all labels associated with the edge.
	
	We say that a hyperedge $e = (v_1,v_2,\ldots,v_h) \in E$ is \emph{consistent} iff there exists $x,y \in [h]$ such that $\pi_{e,v_x}(L(v_x)) \cap \pi_{e,v_y}(L(v_y)) \neq \emptyset$, i.e., $e$ is consistent if there exists two vertices $v_x$ and $v_y$ in this hyperedge such that the projections  of the label sets of $v_x$ and $v_y$ are not disjoint. We will need the following basic lemma which says that the $\ell$-coverage of any inconsistent hyperedge (i.e., a hyperedge which is not consistent) must be small.

	\begin{lemma}				\label{lem:ell-cover-bound}
	
	Let $e = (v_1,v_2,\ldots,v_h) \in E$ be any hyperedge which is inconsistent with respect to $\mathcal{T}$. Then, the $\ell$-coverage of $\mathcal{T}$, restricted to elements $[\widehat{n}]_e$, is upper bounded as follows $C^{(\ell)}_e (\mathcal{T}) \le \left( \fcov^{(\ell)}_{|L(e)|, h} + \eta \right) \widehat{n}$. 
	\end{lemma}
	\begin{proof}
		Since $e$ is inconsistent, we have that for every $x,y \in [h]$, the projected label sets are disjoint i.e., $\pi_{e,v_x}(L(v_x)) \cap \pi_{e,v_y}(L(v_y)) = \emptyset$. Therefore, for every $\alpha_i \in \Sigma$, there exists at most one $v \in e$ such that $\pi_{e,v}(L(v)) \ni \alpha_i$, which implies that for every $i \in [s]$, the family $\mathcal{T}$ intersects with the collection $\mathcal{P}^e_i$ in at most one subset. Therefore, we can invoke condition $(2)$ from Definition~\ref{def:ell-partition} (with $|T| = |L(e)|$), to obtain $C^{(\ell)}_e (\mathcal{T}) \leq \left( \fcov^{(\ell)}_{|L(e)|, h} + \eta \right) \widehat{n}$.  	
		\end{proof}
		
	 Since the overall $\ell$-coverage value of $\mathcal{T}$ is large and inconsistent edges admit small $\ell$-coverage, we can show that there exists a large fraction of consistent edges. To begin with, we claim that for a significant fraction of the edges $e$, the associated label sets $L(e)$ cannot be too large. Specifically, we note that 
	 
	 	\begin{equation}
	\E_{e = (v_1,\ldots,v_h) \sim E}\Big[|L(v_1) \cup L(v_2) \cup \cdots \cup L(v_h)|\Big] \le \frac{h}{|V|}\sum_{v \in V} |L(v)| = h.
	\end{equation}
	
	\noindent Here, we use the fact that the underlying hypergraph is regular and, hence, picking a hyperedge uniformly at random corresponds to selecting vertices with probability $h/|V|$ each.

	Therefore, via Markov's inequality, the size of the label set (i.e., $|L(e)|$) is greater than $ \frac{16  \ell^2}{\delta} \ h $ for at most $\frac{\delta}{16 \ell^2}$ fraction of the hyperedges. 
	Next, we provide a decoding which weakly satisfies a significant fraction of hyperedges which are consistent and have small-sized label sets. 
	
	In particular, we say that a hyperedge $e \in E$ is {\em nice} if (i) $e$ is consistent and (ii) the associated label set $L(e) $ is of cardinality at most $ \frac{ 16 \ell^2}{\delta} \ h$. Since inconsistent edges with small label sets must result in small $\ell$-coverage (Lemma \ref{lem:ell-cover-bound}), it must be that a significant fraction of edges must be nice. This observation is formalized in the following lemma:

	\begin{lemma}
	At least $\frac{\delta}{16 \ell^2}$ fraction of hyperedges must be nice.
	\end{lemma}
	\begin{proof}
		Assume not, for the purpose of contradiction. Then, applying a union bound gives us
	\begin{equation*}
	\Pr\big[ e \mbox { is consistent }\big] \le \Pr\big[ e \mbox{ is nice }\big] + \Pr \left[ |L(e)| \ge \frac{16 \ell^2}{\delta} h  \right] \leq \frac{ \delta}{8 \ell^2} 
	\end{equation*}
	\noindent where the tail bound on $|L(e)|$ follows from Markov's inequality (as mentioned above). Therefore, the $\ell$-coverage contribution from consistent edges is at most $\frac{ \delta}{  8 \ell^2 } |E| \cdot \ell \widehat{n} = \frac{ \delta}{ 8 \ell } |\Gamma|$.\footnote{Recall that $|\Gamma| = |E| \widehat{n}$.} Furthermore, using equation \eqref{ineq:contrap}, we get that the $\ell$-coverage from inconsistent edges is at least $\left( \ell \left(  1-  \frac{\ell^\ell}{\ell!} e^{-\ell}  \right)  + \frac{\delta}{2} \right) |\Gamma| $. 
	
	To obtain a contradiction, we will now prove that the $\ell$-coverage of inconsistent hyperedges cannot be this large. Write  $E_{\rm inc} \subset E$ to denote the set of inconsistent hyperedges. Since $|E_{\rm inc}| \ge (1 - \frac{\delta}{8 \ell^2}) |E|$, by averaging it follows that $\E_{e \sim E_{\rm inc}}[|L(e)|] \le \left(1 + \frac{\delta}{4 \ell^2} \right) h$. 
	
Furthermore, focusing on inconsistent edges and applying Lemma~\ref{lem:ell-cover-bound} we get 
	\begin{align*}
	\E_{e \sim E_{\rm inc}} \left[C^{(\ell)}_e (\mathcal{T})\right] \le \E_{e \sim E_{\rm inc}} \left[{\fcov}^{(\ell)}_{|L(e)|,h} + \eta\right] \
	& \overset{1}{=} \E_{e \sim E_{\rm inc}} \left[\widetilde{\fcov}^{(\ell)}_{|L(e)|,h} + \eta\right] \\
	&\overset{2}{\leq} \left(\widetilde{\fcov}^{(\ell)}_{h(1 + \mu),h} + \eta \right) \widehat{n} \tag{setting $\mu = \frac{\delta}{4 \ell^2}$} \\
	&\overset{3}{=} \left({\fcov}^{(\ell)}_{h(1 + \mu),h} + \eta \right) \widehat{n} \\
	&\leq \left( \ell \left(1 - \frac{\ell^\ell}{\ell!} e^{-\ell} \right) + \delta/4 \right) \widehat{n}
	\end{align*}
	
	Here, in Steps $1$ and $3$ we use the fact that, by construction, $\widetilde{\fcov}^{(\ell)}_{x,h} = \fcov^{(\ell)}_{x,h}$ for every $x \in \mathbbm{Z}_+$. Step $2$ follows via Jensen's inequality along with the observation that $\widetilde{\fcov}^{(\ell)}_{x,h}$ is increasing (Lemma \ref{lem:rho_incr}) and concave in $x$ (Lemma \ref{lem:rho-concave}). Now, the last inequality follows from the last statement of Definition \ref{def:ell-partition}. This implies that the total $\ell$-coverage contribution of inconsistent edges $E_{\rm inc}$ is at most $\left( \ell \left(1 - \frac{\ell^\ell}{\ell!} e^{-\ell} \right) + \delta/4 \right) |\Gamma|$. Hence, we obtain a contradiction and the claim follows.

	\end{proof}

	Finally, we construct a randomized labeling $\sigma: V \mapsto \Sigma$ as follows. For every vertex $v$, if $L(v) \neq \emptyset$, we set $\sigma(v)$ uniformly from $L(v)$, otherwise we set $\sigma(v)$ arbitrarily. We claim that, in expectation, this labeling must weakly satisfy $\Omega(\delta^3)$ fraction of the hyperedges. 
	
	 To see this, fix any nice hyperedge $e = (v_1,v_2,\ldots,v_h)$. Without loss of generality, we can assume that $\pi_{e,v_1}(L(v_1)) \cap \pi_{e,v_2}(L(v_2)) \neq \emptyset$. Furthermore, the niceness also implies that $|L(v_1)|,|L(v_2)| \le \frac{16 \ell^2}{\delta} \ h $. Therefore, with probability at least $1/|L(v_1)||L(v_2)| \ge \frac{\delta^2}{256 \ell^4 h^2}$, we must have $\pi_{e,v_1}(\sigma(v_1)) = \pi_{e,v_2}(\sigma(v_2))$. Therefore,
	\begin{align*}
	\E_\sigma\E_{e \sim E}\left[\mathbbm{1}\{\sigma \mbox{ weakly satisifies } e\}\right] \ge \frac{\delta}{16 \ell^2} \E_\sigma\E_{e \sim E }\left[\mathbbm{1}\{\sigma \mbox{ weakly satisifies } e\} \mid e \in E_{\rm nice} \right] \geq \Omega(\delta^3)	
	\end{align*}
	
	which, for fixed $\ell$ and $h$, gives us the soundness direction.

%% file: convexity-lemma.tex
\section{Proof of Lemma \ref{lem:function-convex} }				\label{sec:function-convex}

\fnconvex*

\begin{proof}
	Let us start with the case $s \geq t$. In this case, 
	\begin{align*}
		\sum_{q=0}^{s-1} (s-q) \binom{t}{q} x^{q} (1 - x)^{t-q} 
		&= \sum_{q=0}^{s} (s-q) \binom{t}{q} x^{q} (1 - x)^{t-q} \\
		&= \sum_{q=0}^{t} (s-q) \binom{t}{q} x^q (1-x)^{t-q} \\
		&= s - \sum_{q=1}^{t} q \binom{t}{q} x^q (1-x)^{t-q} \\
		&= s - \sum_{q=1}^{t} t \binom{t-1}{q-1} x^q (1-x)^{t-q} \\
		&= s - t x \cdot \sum_{q=0}^{t-1} \binom{t-1}{q-1} x^{q} (1-x)^{t-1 - q} \\
		&= s - t x \ ,
	\end{align*}
	which proves the statement.
	
	Now we assume $s < t$. Then we have
	\begin{align*}
		f'(x) = -st (1-x)^{t-1} + \sum_{q=1}^{s-1} (s-q) \binom{t}{q} (q x^{q-1} (1-x)^{t-q} - (t-q) x^q (1-x)^{t-q-1}) \ .
	\end{align*}
	If we now combine the two terms of the form $x^q (1-x)^{t-q-1}$ for $q \in \{0,\dots,s-2\}$, we get a coefficient of 
	\begin{align*}
		-(s-q) (t-q) \binom{t}{q} + (s-q-1) \binom{t}{q+1} (q+1) &= -(s-q) (q+1) \binom{t}{q+1} \\ & \ \ \ \ \ \ + (s-q-1) \binom{t}{q+1} (q+1) \\
		&= - (q+1)\binom{t}{q+1} \ .
	\end{align*}
	Thus, for $x \in [0,1]$
	\begin{align*}
		f'(x) = - \sum_{q=0}^{s-2} \binom{t}{q+1} (q+1) x^{q} (1-x)^{t-q-1} - \binom{t}{s-1} (t-s+1) x^{s-1} (1-x)^{t-s}  \leq 0 \ .
	\end{align*}
	This proves the fact that $f$ is non-increasing. Now, if we differentiate one more time, we get
	\begin{align*}
		f''(x) &= (t-1) \cdot t \cdot (1-x)^{t-2} + \sum_{q=1}^{s-2} \binom{t}{q+1} (q+1) (- q x^{q-1} (1-x)^{t-q-1} + (t-q-1) x^{q} (1-x)^{t-q-2})  \\
		&- \binom{t}{s-1} (t-s+1) (s-1) x^{s-2}(1-x)^{t-s} + \binom{t}{s-1} (t-s+1) (t-s) x^{s-1} (1-t)^{t-s-1} \\
		&= \sum_{q=0}^{s-2} \left( \binom{t}{q+1} (q+1) (t-q-1) - \binom{t}{q+2} (q+2) (q+1) \right) x^{q} (1-x)^{t-q-2} \\
		&+ \binom{t}{s-1} (s-1) (t-s+1) x^{s-2} (1-x)^{t-s} \\
		& - \binom{t}{s-1} (t-s+1) (s-1) x^{s-2}(1-x)^{t-s} + \binom{t}{s-1} (t-s+1) (t-s) x^{s-1} (1-x)^{t-s-1} \\
		&= \binom{t}{s-1} (t-s+1) (t-s) x^{s-1} (1-x)^{t-s-1} \geq 0 \ ,
	\end{align*}
	which proves the convexity of $f$. 
	
\end{proof}

%% file: concavity.tex
\newcommand{\bin}{\rm{Bin}}
\section{Concavity of $\widetilde{\fcov}$}

In this section, we prove that the linear piecewise extension of the function $\fcov^{(\ell)}_{x,h}$ (which we denote by $\widetilde{\fcov}^{(\ell)}_{x,h}$) is concave in $x$. Our strategy would be to show that for a fixed choice of $\ell,h$, the quantity $\fcov^{(\ell)}_{x,h}$ is \emph{increasing} and satisfies a {diminishing marginals} property in $x$. Using these properties, the concavity of $\widetilde{\fcov}$ follows immediately. We setup some additional notation. For the rest of this section, we will fix $\ell$ and $h$, and let $p \coloneqq \ell/h$. Therefore, we will drop by indexing by $\ell$ and $h$, and denote $\fcov(x) = \fcov^{(\ell)}_{x,h}$. 

We will use $Z_{x,p}$ to denote a random variable drawn from the distribution ${\rm Bin}(x,p)$ i.e., the binomial distribution with bias $p$ and number of trials $x$. Let $s: \mathbbm{R}^+ \mapsto \mathbb{R}$ denote the piecewise linear function defined as follows: $s(x) \coloneqq x$ for all $x \le \ell$ and $s(x) = \ell $ for all $x > \ell$. Recall that from the proof of Lemma \ref{lem:gadget} we can alternatively write $\fcov(x)$ as 
\begin{equation}			\label{def:rho}
\fcov(x) = \sum_{i = 0}^{x} s(i){x \choose i} p^i(1 - p)^{x - i} = \E\Big[s\big(Z_{x,p}\big)\Big]
\end{equation}
Going forward, the above expression for $\fcov(x)$ will prove to be useful. We begin by the following lemma, which says that $\fcov(x)$ is increasing in $x$.

\begin{lemma}					\label{lem:rho_incr}
	For every $x \in \mathbbm{Z}_+$ we have $\fcov({x+1}) \ge \fcov(x)$.
\end{lemma}
\begin{proof}
This is a direct consequence of stochastic dominance between binomial distributions with $x$ and $x+1$ trials, respectively. Recall that $\fcov(x) = \E\Big[s\big(Z_{x,p}\big)\Big]$. The tail-sum formula gives us  
\begin{equation*}
\E\Big[s\big(Z_{x,p}\big)\Big] = \sum_{t = 0}^\infty \Pr\left[ s(Z_{x,p}) \ge t\right] = \sum_{t = 0}^\ell \Pr\left[ s(Z_{x,p}) \ge t\right] = \sum_{t = 0}^\ell \Pr\left[Z_{x,p} \ge t\right]. 
\end{equation*}   
Here, the second equality follows from the fact that, by definition, $s(z) \leq \ell$ for all $z \in \mathbbm{Z}_+$. The third inequality relies on the observation that, for any $t \in [\ell]$, we have $s(z) \geq t$ iff $z \geq t$. 
Analogously, we have $\fcov(x+1) =  \sum_{t = 0}^\ell \Pr\left[Z_{x+1,p} \ge t\right]$. 

Since the random variable $Z_{x+1, p}$ stochastically dominates $Z_{x,p}$---in particular, $ \Pr\left[Z_{x+1,p} \ge t\right] \geq \Pr\left[Z_{x,p} \ge t\right]$---the desired inequality follows $\fcov(x+1) \geq \fcov(x)$. 

\end{proof}

Additionally, the following lemma proves that $\fcov$ satisfies a diminishing marginals property.
	
\begin{lemma}					\label{lem:incr_marginal}
	For every $x,y \in \mathbbm{Z}_+$ such that $x \le y$ we have $\fcov({x+1}) - \fcov(x) \ge \fcov({y+1}) - \fcov(y)$.
\end{lemma}
\begin{proof}
	We again use the alternative expression for $\fcov$ as given in (\ref{def:rho}):
	\begin{eqnarray*}
	\fcov({x+1}) = \E\Big[s\big(Z_{x+1,p}\big)\Big] &=& \sum_{i \ge 0}  \E\Big[s\big( i + Z_{1,p} \big) \Big| Z_{x,p} = i\Big] \cdot \Pr\Big[Z_{x,p} = i \Big] \\
	&=& \sum_{i \ge 0}  \E\Big[s\big( i + Z_{1,p} \big) - s(i) + s(i) \Big| Z_{x,p} = i\Big] \cdot \Pr\Big[Z_{x,p} = i \Big] \\
	&=& \sum_{i \ge 0}  \E\Big[s\big(Z_{1,p}\big) - s(i) \Big| Z_{x,p} = i\Big] \cdot \Pr\Big[Z_{x,p} = i\Big]  \\ & \ & + \sum_{i \ge 0} \E\Big[s(i) \Big| Z_{x,p} = i\Big] \cdot  \Pr\Big[Z_{x,p} = i\Big]   \\
	&=& \sum_{i \ge 0}  \E\Big[s\big(i + Z_{1,p}\big) - s(i) \Big| Z_{x,p} = i\Big] \cdot \Pr\Big[Z_{x,p} = i\Big] \ + \ \E\Big[s\big(Z_{x,p}\big)\Big]   \\
	&=& \sum_{i \ge 0} \E\Big[s\big(i + Z_{1,p}\big) - s(i) \Big| Z_{x,p} = i\Big] \cdot \Pr\Big[Z_{x,p} = i\Big]  \ + \ \fcov(x)  
	\end{eqnarray*}


Rearranging the above expression we get  
	\begin{eqnarray*}
	\fcov({x+1}) - \fcov(x) &=& \sum_{i \ge 0} \E\Big[ s\big(i + Z_{1,p}\big) - s(i) \Big] \Pr\Big[Z_{x,p} = i\Big]  \\   
	&\overset{1}{=}& \sum_{i=0}^{\ell-1}  \E\Big[ s\big(i + Z_{1,p}\big) - s(i) \Big] \Pr\Big[Z_{x,p} = i\Big] \\
	&\overset{2}{=}& \sum_{i=0}^{\ell-1}  \E\Big[Z_{1,p}\Big] \Pr\Big[Z_{x,p} = i\Big] \\
	&{=}& \sum_{i=0}^{\ell-1} p \Pr\Big[Z_{x,p} = i\Big]  \\
	&{=}& p \Pr\Big[Z_{x,p} \le \ell - 1\Big]  \\
	&\overset{3}{\ge}& p \Pr\Big[Z_{y,p} \le \ell - 1\Big]  \\
	&\overset{4}{=}& \fcov({y+1}) - \fcov(y)
	\end{eqnarray*}

Steps $1$ and $2$ follow from the construction of the function $s$; in particular, for any $0 \le z \le 1$, we have $s(i + z) = s(i)$ whenever $i \ge \ell$ and $s(i + z) = s(i) + z$ if $i \le \ell - 1$. Step $3$ follows from the stochastic dominance of $Z_{y,p}$ over $Z_{x,p}$; recall that $x \leq y$. Step $4$ is obtained by reapplying the arguments used till Step $3$ for $y$, i.e., by instantiating the expressions with $y$, instead of $x$.  
\end{proof}

Now we are ready to prove the concavity of $\widetilde{\fcov}^{(\ell)}_{x,h}$.

\begin{lemma}					\label{lem:rho-concave}
	Let $\widetilde{\fcov}^{(\ell)}_{x,h} $ be the piecewise linear extension of $\fcov^{(\ell)}_{x,h}$. Then, $\widetilde{\fcov}^{({\ell})}_{x,h}$ is concave in $x$.
\end{lemma}
\begin{proof}
	For ease of notation, we drop the indexing by $\ell,h$ and denote $\widetilde{\fcov}(x) = \widetilde{\fcov}^{(\ell)}_{x,h}$. We shall need the following equivalent characterization of concavity of function:
	
	\begin{proposition}					\label{prop:cncv}
		A function $g:A \mapsto \mathbbm{R}$ is concave over $A \subseteq \mathbbm{R}$ iff for every choice of $x_1,x_2,x_3,x_4 \in A$ such that $x_1 \le x_2 \le x_3 \le x_4$ we have 
		\begin{equation}
		\frac{g(x_2) - g(x_1)}{x_2 - x_1} \ge \frac{g(x_4) - g(x_3)}{x_4 - x_3}
		\end{equation}  
	\end{proposition}
	
	Using the above proposition, the concavity of $\widetilde{\fcov}$ follows almost directly. Fix any ${x_i}_{i \in [4]}$ as in Proposition \ref{prop:cncv}. For every $j \in \{1,2,3,4,\}$, let $(i_j,i_j + 1]$ be the semi-closed interval such that $x_j \in (i_j,i_j + 1]$. Then,
	
	\begin{eqnarray}
	\frac{\widetilde{\fcov}(x_2) - \widetilde{\fcov}(x_1)}{x_2 - x_1}  \label{eq:ineq1}
	&=& \frac{(\widetilde{\fcov}(x_2) - \widetilde{\fcov}(i_2)) +  \sum_{r = i_1+2}^{i_2} (\widetilde{\fcov}(r) - \widetilde{\fcov}(r-1))  + (\widetilde{\fcov}(i_1 + 1) - \widetilde{\fcov}(x_1)) }
	{(x_2 - i_2 ) + \sum_{r = i_1+2}^{i_2} (r - (r-1))  + ((i_1 + 1) - x_1) } \\ \nonumber
	 &\overset{1}{\ge}& \max\left\{\frac{\widetilde{\fcov}(x_2) - \widetilde{\fcov}(i_2)}{x_2 - i_2}, \max_{r \in \{i_1+2,\ldots,i_2\}} \Big(\frac{\widetilde{\fcov}(r) - \widetilde{\fcov}(r-1)}{r - (r-1)} \Big)  , \frac{\widetilde{\fcov}(i_1 + 1) - \widetilde{\fcov}(x_1)}{(i_1 + 1) - x_1}\right\} \\ \nonumber
	&\overset{2}{\ge}& \max\left\{\frac{\widetilde{\fcov}(x_2) - \widetilde{\fcov}(i_2)}{x_2 - i_2}, \frac{\widetilde{\fcov}(i_1 + 2) - \widetilde{\fcov}(i_1 + 1)}{(i_1 + 2) - (i_1 + 1)}  , \frac{\widetilde{\fcov}(i_1 + 1) - \widetilde{\fcov}(x_1)}{(i_1 + 1) - x_1}\right\}\\ \nonumber
	&\overset{3}{=}& \max\left\{\frac{\widetilde{\fcov}(i_2 + 1) - \widetilde{\fcov}(i_2)}{(i_2 + 1) - i_2}, \frac{\widetilde{\fcov}(i_1 + 2) - \widetilde{\fcov}(i_1 + 1)}{(i_1 + 2) - (i_1 + 1)}  , \frac{\widetilde{\fcov}(i_1 + 1) - \widetilde{\fcov}(i_1)}{(i_1 + 1) - i_1}\right\}\\ 
	&\overset{4}{=}& \frac{\widetilde{\fcov}(i_2 + 1) - \widetilde{\fcov}(i_2)}{(i_2 + 1) - i_2}  \label{eq:ineq2}
	\end{eqnarray} 
	
	We briefly justify the above steps. Step $1$ uses the following known observation that for any sequence of pairs of nonnegative integers $(a_i,b_i)_{i \in [r]}$ we have 
	
	\begin{equation}
	\frac{a_1 + a_2 + \cdots + a_r}{b_1 + b_2 + \cdots + b_r} \ge \min_{i \in [r]} \frac{a_i}{b_i}
	\end{equation}
	
	Combining the above observation with the fact that $\widetilde{\fcov}(x)$ is increasing in $x$ gives us the inequality. In steps $2$ and $4$, we use the diminishing marginal property of $\widetilde{\fcov}$ (Lemma \ref{lem:incr_marginal}). Step $3$ follows from the piecewise linearity of $\widetilde{\fcov}$. A similar sequence of arguments also gives us
	
	\begin{equation}					\label{eq:ineq3}
	\frac{\widetilde{\fcov}(x_4) - \widetilde{\fcov}(x_3)}{x_4 - x_3} {\le} \frac{\widetilde{\fcov}(i_3 + 1) - \widetilde{\fcov}(i_3)}{(i_3 + 1) - i_3} 
	\end{equation}
	
	Since $x_3 \ge x_2$, we have $i_3 \ge i_2$, and therefore, using the diminishing marginals property of $\widetilde{\fcov}$ (Lemma \ref{lem:rho_incr}) we have 
	
	\begin{equation}					\label{eq:ineq4}
	\frac{\widetilde{\fcov}(i_3 + 1) - \widetilde{\fcov}(i_3)}{(i_3 + 1) - i_3} \ge \frac{\widetilde{\fcov}(i_3 + 1) - \widetilde{\fcov}(i_3)}{(i_3 + 1) - i_3}
	\end{equation}	
	
	Combining the inequalities from Equations (\ref{eq:ineq1}-\ref{eq:ineq2}), (\ref{eq:ineq3}) and (\ref{eq:ineq4}) gives us $\frac{\widetilde{\fcov}(x_2) - \widetilde{\fcov}(x_1)}{x_2 - x_1} \ge \frac{\widetilde{\fcov}(x_4) - \widetilde{\fcov}(x_3)}{x_4 - x_3}$. Since this holds for any choice of $x_1 \le x_2 \le x_3 \le x_4$, using Proposition \ref{prop:cncv}, we get that $\widetilde{\fcov}$ is concave.
	\end{proof}

%% file: hypergraph-ugc.tex
 \section{Reduction to $h$-\kary} \label{sec:hyper-ugc}

Here, we give the reduction from the graph variant of the $\ug$ to the $h$-$\kary$ that we use in our reduction. We point out that this variant is well known, and in particular, a near identical variant can be found in \cite{feldman2012agnostic}. However the variant from \cite{feldman2012agnostic} does not explicitly guarantee that the underlying constraint hypergraph is regular, a feature we use crucially in our reduction. Hence, we include the full reduction for the sake of completeness. We begin by introducing the conjecture for bi-regular variant of $\ug$.

\begin{definition}[$\ug$]
		An instance $\mathcal{G}(U,V,E,\Sigma,\{\pi_{e,u}:\Sigma \mapsto \Sigma\}_{e\in E, v \in V})$ of $\ug$ is characterized by a bipartite graph on vertices $(U,V)$ and bijection projection constraints $\pi_{e,v}:\Sigma \mapsto \Sigma$. Here, each edge represents a constraint involving the vertices participating in the edge. We say that a labeling $\sigma:U \cup V \mapsto \Sigma$ satisfies the edge $(u,v)  \in E$ if and only if $\pi_{e,v}(\sigma(v)) = \sigma(u)$.
\end{definition}
		
The following is known to be equivalent to the Unique Games Conjecture

\begin{conjecture}[See Conjecture 1 \cite{bhangale-biUG}	]		\label{conj:ug}
	For every constant $\epsilon > 0$ the following holds. Given an instance $\mathcal{G}$ of $\ug$, it is NP-Hard to distinguish between the following cases:
	\begin{itemize}
		\item (YES): There exists a labeling $\sigma$ of the vertices which satisfies at least $1-\epsilon$ fraction of the constraints.
		\item (NO): No labeling $\sigma$ of the vertices satisfies more than $\epsilon$ fraction of the edges. 
	\end{itemize}
Additionally, the underlying constraint graph is regular. Here the degree of the constraint graph and the alphabet size depend only on the parameter $\epsilon$.
\end{conjecture}

The following theorem says that there exists a polynomial time reduction from $\ug$ to $h$-$\kary$.

\begin{theorem}							\label{thm:hyper-ugc}
	For all constant choices of $\epsilon > 0$ and $h \in \mathbbm{N}$, there exists a polynomial time reduction which on input a $\ug$ instance $\mathcal{G}(U,V,E,\Sigma,\{\pi_{e,u}\}_{e \in E,u\in e})$ (as in Conjecture \ref{conj:ug}) outputs $h$-$\kary$ instance $\mathcal{G}'(V,E',\Sigma,\{\tilde{\pi}_{e,u}\}_{e \in E',u \in e})$ satisfying the the following properties: 
	\begin{itemize}
		\item If $\mathcal{G}$ is a YES instance, then there exists a labeling which strongly satisfies $1 - \epsilon$ fraction of hyperedges in $\mathcal{G}'$.
		\item If $\mathcal{G}$ is a NO instance, then no labeling weakly satisfies more than $h^2\sqrt{\epsilon}$ fraction of the hyperedges in $\mathcal{G}'$.
	\end{itemize}   
	Additionally, the instance $\mathcal{G}'$ output by the reduction satisfy the following properties.
	\begin{itemize}
		\item The alphabet set of $\mathcal{G}'$ is the same as the alphabet set of $\mathcal{G}$.
		\item The underlying constraint hypergraph is {\em regular} i.e., every vertex $v \in V$ participates in the same number of hyperedge constraints.
	\end{itemize}
\end{theorem}

\begin{proof}
	We construct the $h$-$\kary$ instance as follows. The vertex set of the $h$-$\kary$ instance $\mathcal{G}'$ is going to be $V$ i.e., the right vertex set of the $\ug$ instance $\mathcal{G}$. The underlying constraint hypergraph is the following $h$-ary hypergraph $H$. Fix a left vertex $u \in U$, and let $N_(u)$ denote its neighborhood. For every $h$-sized subset $(v_1,v_2,\ldots,v_h) \subset {N(u) \choose h}$ we add the hyperedge $e = (v_1,v_2,\ldots,v_h)$ to the hyperedge set $E'$. Furthermore, we set the corresponding bijection constraint to be $\tilde{\pi}_{e,v_i} = \pi_{(u,v_i),v_i}$. This is done for every choice of left vertex $u$, and every $h$-sized subset of its neighborhood. Overall, for a constant $h$, the reduction runs in time $|V|^{O(h)} {\rm poly}(|U|,|V|,|\Sigma|)$.
	  
	From its construction, it is clear that every vertex in $v$ participates in the same number of $h$-ary constraints (this follows from the bi-regularity of the original constraint graph). Furthermore, we shall need the following observation which is again a consequence of the bi-regularity of the $\ug$ instance $\mathcal{G}$.
	
	{\bf Observation 1} The following process is an equivalent way of sampling a random hyperedge $e \in E'$.
	
	\begin{itemize}
		\item Sample a random left vertex $u \sim U$.
		\item Sample $h$-random neighbors $v_1,v_2,\ldots,v_h \sim N(u)$ without replacement, and output the hyperedge $e = (v_1,v_2,\ldots,v_h)$
	\end{itemize}
	
	Equipped with the above observations, we now argue the completeness and soundness directions of our reduction. 
	
	{\bf Completeness}: Suppose $\mathcal{G}$ is a YES instance. Then, there exists a labeling $\sigma:U \cup V \mapsto \Sigma$ of the vertices which satisfies at least $1 - \epsilon$ fraction of the edges. Let $\sigma':V \mapsto \Sigma$ be the restriction of the labeling $\sigma$ to the set of right vertices $V$ i.e., for all $v \in V$ we have $\sigma'(v) = \sigma(v)$. We now show that the labeling $\sigma'$ strongly satisfies at least $1 - \epsilon h$ fraction of hyperedges in $\mathcal{G}'$.
	\begin{align*}
		&\Pr_{e = (v_1,v_2,\ldots,v_h) \sim E'}\Big[\sigma' \mbox{strongly satsifies } e \Big]\\
		&=\E_{u\sim U}\Bigg[\Pr_{v_1,v_2,\ldots,v_h \sim N(u)} \Big[\forall {i \neq j}, \tilde{\pi}_{e,v_i}(\sigma'(v_i)) = \tilde{\pi}_{e,v_j}(\sigma'(v_j)) \Big]\Bigg] \tag{Observation 1}\\
		&\ge \E_{u\sim U}\Bigg[ \Pr_{v_1,v_2,\ldots,v_h \sim N(u)} \Big[\forall i \in [h], \tilde{\pi}_{e,v_i}(\sigma'(v_i)) = \sigma(u) \Big] \Bigg]\\
		&= \E_{u\sim U}\Bigg[\Pr_{v_1,v_2,\ldots,v_h \sim N(u)} \Big[\forall i \in [h], {\pi}_{(u,v_i),v_i}(\sigma(v_i)) = \sigma(u) \Big]\Bigg]\\
		&\ge 1 - \sum_{i \in [h]}\Pr_{u \sim U,v_i \sim N(u)} \Big[{\pi}_{(u,v_i),v_i}(\sigma(v_i)) \neq \sigma(u) \Big]\\
		&\ge 1 - \epsilon h
	\end{align*}
	The inequality in the last step can be justified as follows. We use the fact that for $u$ drawn uniformly random from $U$ and $v_i$ drawn uniformly random from $N(u)$, the pair $(u,v_i)$ is marginally distributed as a uniformly random edge from $E$ (since $\mathcal{G}$ is bi-regular). Since the labeling $\sigma$ satisfies at least $1-\epsilon$ fraction of edges,  for each $i \in [h]$, the probability that $\sigma$ does not satisfy the edge $(u,v_i)$ is at most $\epsilon$. Combining the two observations gives us the inequality.
	 
	{\bf Soundness}: Suppose in the $h$-$\kary$ instance, there exists a labeling of the vertices $\sigma':V \mapsto \Sigma$ which weakly satisfies at least $\epsilon$ fraction of the hyperedges in $\mathcal{G}'$. Using Observation $1$ and the construction of $\tilde{\pi}$, this is equivalent to  
	
	\begin{equation}
	\Pr_{u,\{v_1,v_2,\ldots,v_h\} \sim N(u)}\left[ \exists i \neq j \mbox{ s.t. }  {\pi}_{(u,v_i),v_i}(\sigma'(v_i)) = {\pi}_{(u,v_j),v_j}(\sigma'(v_j))\right] \ge \epsilon
	\end{equation}
	
	By averaging over the choices of pairs of indices, there exists indices $i,j \in [h], i \neq j$, such that 
	
	\begin{equation}
	\E_{u \sim U} \Bigg[\Pr_{v_i,v_j \sim N(u)}\left[ {\pi}_{(u,v_i),v_i}(\sigma'(v_i)) = {\pi}_{(u,v_j),v_j}(\sigma'(v_j))\right] \Bigg] \ge \frac{\epsilon}{{h\choose 2}} \ge 2\epsilon/h^2
	\end{equation}
	
	Again by an averaging argument, we know that for at least $\epsilon/h^2$ choices of left vertices $u$, we have 
	  
	\begin{equation}
	\Pr_{v_i,v_j \sim N(u)}\left[ {\pi}_{(u,v_i),v_i}(\sigma'(v_i)) = {\pi}_{(u,v_j),v_j}(\sigma'(v_j))\right] \ge \epsilon/h^2
	\end{equation}
	
	We call such a left vertex $u$ as \emph{good}. Then, for any fixed good vertex $u \in U$, there exists a right vertex $v(u) \in V$, for which 
	
	\begin{equation}
	\Pr_{v_i \sim N(u)}\left[ {\pi}_{(u,v(u)),v(u)}(\sigma'(v(u))) = {\pi}_{(u,v_i),v_i}(\sigma'(v_i))\right]  \ge \epsilon/h^2
	\end{equation}
	
	In other words, for at least $\epsilon/h^2$ fraction of right vertices of $v' \in N(u)$, the corresponding right vertex label $\sigma'(v')$ projects under $\pi_{(u,v'),v'}$ to the same left vertex label, say $\sigma_u \in \Sigma$. In particular, we denote the set of all right vertices $v' \in V'$ which project to $\sigma_u$ as $V(u)$. We shall use these left vertex labels and the labeling $\sigma'$ to construct a labeling $\sigma: U \uplus V \mapsto \Sigma$ which shall satisfy a significant fraction of edges in $\mathcal{G}$. Specifically, for every good vertex $u \in U$, we assign $\sigma(u) = \sigma_u$.  We complete the labeling of left vertices by assigning labels to unlabeled vertices arbitrarily. The right vertices are labeled exactly according to the labeling $\sigma'$. We now show that this labeling satisfies at least $\epsilon^2/h^4$-fraction of edges in the $\ug$ instance $\mathcal{G}$. 
	
	\begin{eqnarray*}
	\E_{e = (u,v)} \Big[\mathbbm{1}_{\{\sigma \mbox{\small{ satisfies} } e\}}\Big]
	&=&\E_{e = (u,v)}\Big[\mathbbm{1}_{\{\pi_{e,v}(\sigma(v)) = \sigma(u)\}}\Big] \\
	&=& \E_{u \in U}\E_{v \sim N(u)}\Big[\mathbbm{1}_{\{\pi_{(u,v),v}(\sigma(v)) = \sigma(u)\}}\Big] \\
	&\ge& \Pr_{u \sim L}\left[ u \mbox{ is good} \right]\E_{v \sim N(u)}\Big[\mathbbm{1}_{\{\pi_{(u,v),v}(\sigma(v)) = \sigma(u)\}} | u \mbox{ is good}\Big] \\
	&\ge& \frac{\epsilon}{h^2}\E_{u \sim U,v \sim N(u)}\Big[\mathbbm{1}_{\{\pi_{(u,v),v}(\sigma(v)) = \sigma_u\}} | u \mbox{ is good}\Big] \\
&\ge& \frac{\epsilon}{h^2}\E_{u \sim U,v \sim N(u)}\Big[\Pr[v \in V(u)]\mathbbm{1}_{\{\pi_{(u,v),v}(\sigma(v)) = \sigma_u\}} | u \mbox{ is good},v \in V(u)\Big] \\
&\ge& \frac{\epsilon^2}{h^4}\E_{u \sim U,v \sim N(u)}\Big[\mathbbm{1}_{\{\pi_{(u,v),v}(\sigma(v)) = \sigma_u\}} | u \mbox{ is good},v \in V(u)\Big] \\
&=& \frac{\epsilon^2}{h^4}
	\end{eqnarray*}
	
	where in the first step we use the fact that the $\ug$ instance $\mathcal{G}$ is bi-regular, and hence left regular. In the last step we know that for any choice of $u \in U$ such that $u$ is good, and any choice of $v \in V(u)$, we must have $\sigma_u = \pi_{(u,v),v}(\sigma(v))$. This completes the proof of soundness direction.
\end{proof}	
	
\section{Increasing Alphabet Size While Preserving Completeness and Soundness}		\label{sec:alphabet}
	
In this section, we state and prove the following lemma which shows that one can always choose alphabet size of the $\ug$ instance to be large enough while preserving the completeness and soundness parameters of the instance.

\begin{lemma}
	Let $\mathcal{G}(V,E,\Sigma,\{\pi_{e,v}\}_{e,v})$ be an instance of $h$-$\kary$. Let $r \in \mathbbm{N}$ be a nonnegative integer. Then there exists a polynomial time algorithm which constructs  a $h$-$\kary$ instance $\mathcal{G}'(V,E,\Sigma',\{\pi'_{e,v}\}_{e,v})$ such that $|\Sigma'| = r|\Sigma|$ satisfying the following property for any choice of $\gamma \in [0,1]$. There exists a labeling $\sigma:V \mapsto \Sigma$ which strongly (or weakly) satisfies at least $\gamma$ fraction of the hyperedges in $\mathcal{G}$ {\em iff} there exists a labeling $\sigma':V \mapsto \Sigma'$ which strongly (or weakly) satisfies at least $\gamma$ fraction of hyperedges in $\mathcal{G}'$.
\end{lemma}
\begin{proof}
	The underlying idea here is to define the large alphabet set $\Sigma'$ for the $h$-$\kary$ instance $\mathcal{G}'$ to be the disjoint union of $r$ copies of the smaller alphabet set $\Sigma$. Specifically, we define $\Sigma' = \Sigma_1 \uplus \Sigma_2 \uplus \ldots \uplus \Sigma_r$, where $\Sigma_1,\Sigma_2,\ldots,\Sigma_r$ are disjoint copies of the alphabet set $\Sigma$. Let $\Sigma: = \{\alpha_1,\alpha_2,\ldots,\alpha_s\}$, and for every $i \in [r]$, let $\Sigma_i = \{\alpha^{(i)}_1,\alpha^{(i)}_2,\ldots,\alpha^{(i)}_s\}$. As stated in the lemma, the vertex set and the hyperedge set for $\mathcal{G}'$ are the same as those of $\mathcal{G}$. Additionally, for any $h$-ary hyperedge $e$ and any vertex $v \in E$, we define the projection constraint $\pi'_{e,v}:\Sigma' \mapsto \Sigma'$ as follows. For every $i \in [r]$, the restriction of the projection $\pi'_{e,v}$ to the alphabet set $\Sigma_i$ is the corresponding copy of $\pi_{e,v}:\Sigma \mapsto \Sigma$ on the alphabet set $\Sigma_i$. Formally, for every index $i \in [r]$, and label $\alpha^{(i)}_j \in \Sigma_i$, and $e \in E', v \in e$, we assign $\pi_{e,v}(\alpha^{(i)}_j) = \alpha^{(i)}_{j'}$ if $\pi_{e,v}(\alpha_j) = \alpha_{j'}$. This completes the description of the $h$-$\kary$ instance $\mathcal{G}'$. Note that we can construct $\mathcal{G}'$ from $\mathcal{G}$ in time ${\rm poly}(|V|,|\Sigma|,r)$.  
	
	We shall prove the claim for strongly satisfied hyperedges; the case involving the weakly satisfied hyperedges follows similarly. We begin by arguing the forward direction of the claim. Suppose there exists labeling $\sigma:V  \mapsto \Sigma$ which strongly satisfies at least $\gamma$ fraction of the hyperedges in $\mathcal{G}$. Then, we construct labeling $\sigma':V \mapsto \Sigma'$ for $\mathcal{G}'$ from $\sigma$ as follows. For every $v \in V$, we let $\sigma'(v) = \alpha^{(1)}_{i(v)}$ if $\sigma(v) = \alpha_{i(v)}$. We claim that the labeling $\sigma'$ will strongly satisfy any hyperedges strongly satisfied by $\sigma$ in $\mathcal{G}$. To see this, we fix a hyperedge $e=(v_1,v_2,\ldots,v_h)$ (which are w.l.o.g., labeled with $\alpha_1,\alpha_2,\ldots,\alpha_h$ respectively) which is strongly satisfied by $\sigma$. Then $\pi_{e,v_1}(\sigma(v_1)) = \pi_{e,v_2}(\sigma(v_2)) = \cdots = \pi_{e,v_h}(\sigma(v(h))) = \alpha_{i(e)}$ for some $\alpha_{i(e)} \in \Sigma$. But then, our construction of $\Sigma'$ ensures that for every $j \in [h]$ we have 
	\begin{equation}
	\pi'_{e,v_j}(\sigma'(v_j)) = \pi'_{e,v_j}(\alpha^{(1)}_{i(v_j)}) = \alpha^{(1)}_{i(e)}
	\end{equation} 
	The above implies that the labeling $\sigma'$ strongly satisfies the hyperedge $e$ in the $h$-$\kary$ instance $\mathcal{G}'$. Note that these observations hold for any hyperedge $e \in E$ strongly satisfied by the labeling $\sigma$ in the $h$-$\kary$ instance $\mathcal{G}$. Therefore, the labeling $\sigma'$ must also strongly satisfy at least $\gamma$-fraction of the hyperedges in $\mathcal{G}'$.
	
	Now we prove the reverse direction of the claim. Let $\sigma':V \mapsto \Sigma'$ be a labeling strongly satisfying at least $\gamma$ fraction of the hyperedges in $\mathcal{G}'$. We shall construct from $\sigma'$ a labeling $\sigma: V \mapsto \Sigma$ of the vertices of $\mathcal{G}$ which will strongly satisfy at least $\gamma$ fraction of the hyperedges in $\mathcal{G}$. Formally, for any vertex $v \in V$, we assign $\sigma(v) = \alpha_{i(v)}$ if $\sigma'(v) \in \{\alpha^{(w)}_{i(v)} | w \in [r]\}$. As in the previous part, we claim that if the labeling $\sigma'$ strongly satisfies a hyperedge $e = (v_1,v_2,\ldots,v_h)$ in $\mathcal{G}'$, then the labeling $\sigma$ must strongly satisfy the hyperedge $e$ in $\mathcal{G}$. For $j \in [h]$ let $\sigma'(v_j) = \alpha^{(w_j)}_{i(v_j)}$ be the label assigned to the $j^{th}$ vertex in $e$. Then by construction of $\sigma$, for any vertex $v_j \in e$ we have $\sigma(v_j) = \alpha_{i(v_j)}$. Since $\sigma'$ strongly satisfies the hyperedge $e \in E$, we have $\pi'_{e,v_1}(\sigma'(v_1)) = \pi'_{e,v_2}(\sigma'(v_2)) = \cdots = \pi'_{e,v_h}(\sigma'(v_h)) = \alpha^{w_e}_{i_e}$ for some $\alpha^{w_e}_{i_e} \in \Sigma_{w_e}$. But then by construction of $\sigma$, for every choice of $j \in [h]$ we must have $\pi_{e,v_j}(\sigma(v_j)) = \pi_{e,v_j}(\alpha_{i(v_j)}) = \alpha_{i_e}$, which implies that $\sigma$ strongly satisfies the hyperedge $e$ in $\mathcal{G}$. Again, note that this holds for any hyperedge $e$ strongly satisfied by the labeling $\sigma'$ in $\mathcal{G}'$. Since $\sigma'$ strongly satisfies at least $\gamma$ fraction of the hyperedges in $\mathcal{G}'$, we can conclude that the labeling $\sigma$ strongly satisfies at least $\gamma$ fraction of the hyperedges in $\mathcal{G}$.
	
	\end{proof}